\documentclass[11pt]{article}

\usepackage{amsmath,amssymb,amsthm,verbatim,relsize,mathrsfs,hyperref,bbm}
\usepackage{algorithm}
\usepackage[noend]{algpseudocode}
\hypersetup{colorlinks = true,linkcolor = blue, anchorcolor = blue, citecolor = blue, filecolor = red,urlcolor = red}
\usepackage [english]{babel}
\usepackage [autostyle, english = american]{csquotes}
\usepackage{blindtext}
\usepackage{framed}
\usepackage{fullpage}
\usepackage{paralist}
\usepackage{enumitem}
\usepackage{graphicx}
\usepackage{subcaption}
\usepackage{mathtools}
\usepackage{times}

\MakeOuterQuote{"}

\makeatletter
\newcommand\footnoteref[1]{\protected@xdef\@thefnmark{\ref{#1}}\@footnotemark}
\makeatother

\newcommand{\Expect}{\mathbf{E}}

\newcommand{\ignore}[1]{}

\newcommand{\cD}{\mu}

\newcommand{\cF}{\mathcal{F}}

\newcommand{\cI}{{\mathcal{I}}}

\newcommand{\cL}{{\mathcal{L}}}
\newcommand{\cM}{{\mathcal{M}}}
\newcommand{\cP}{\mathcal{P}}

\newcommand{\cS}{\mathcal{S}}

\newcommand{\eps}{\varepsilon}

\newcommand{\poly}{\mathrm{poly}}

\newcommand{\cO}{{\cal O}}

\newcommand{\calB}{{\cal B}}

\newcommand{\bo}{\mathsf{proj}_1}

\newcommand{\bB}{\mathbf{B}}
\newcommand{\bG}{\mathbf{G}}

\newcommand{\bI}{\boldsymbol{I}}
\newcommand{\bJ}{\boldsymbol{J}}

\newcommand{\bT}{\boldsymbol{T}}

\newcommand{\bZ}{\boldsymbol{Z}}

\newcommand{\NN}{\mathbb{N}}
\newcommand{\RR}{\mathbb{R}}

\newcommand{\ceil}[1]{\lceil#1\rceil}
\newcommand{\Exp}{\EX}

\newcommand{\EX}{\hbox{\bf E}}

\newcommand{\Sec}[1]{\hyperref[sec:#1]{\S\ref*{sec:#1}}} 
\newcommand{\Eqn}[1]{\hyperref[eq:#1]{(\ref*{eq:#1})}} 
\newcommand{\Fig}[1]{\hyperref[fig:#1]{Fig.\,\ref*{fig:#1}}} 
\newcommand{\Tab}[1]{\hyperref[tab:#1]{Tab.\,\ref*{tab:#1}}} 
\newcommand{\Thm}[1]{\hyperref[thm:#1]{Theorem\,\ref*{thm:#1}}} 
\newcommand{\Fact}[1]{\hyperref[fact:#1]{Fact\,\ref*{fact:#1}}} 
\newcommand{\Lem}[1]{\hyperref[lem:#1]{Lemma\,\ref*{lem:#1}}} 
\newcommand{\Prop}[1]{\hyperref[prop:#1]{Prop.~\ref*{prop:#1}}} 
\newcommand{\Cor}[1]{\hyperref[cor:#1]{Corollary~\ref*{cor:#1}}} 
\newcommand{\Conj}[1]{\hyperref[conj:#1]{Conjecture~\ref*{conj:#1}}} 
\newcommand{\Def}[1]{\hyperref[def:#1]{Definition~\ref*{def:#1}}} 
\newcommand{\Alg}[1]{\hyperref[alg:#1]{Alg.~\ref*{alg:#1}}} 
\newcommand{\Ex}[1]{\hyperref[ex:#1]{Example~\ref*{ex:#1}}} 
\newcommand{\Clm}[1]{\hyperref[clm:#1]{Claim~\ref*{clm:#1}}} 
\newcommand{\Inv}[1]{\hyperref[inv:#1]{Invariant~\ref*{inv:#1}}} 
\newcommand{\Rem}[1]{\hyperref[rem:#1]{Remark~\ref*{rem:#1}}} 
\newcommand{\Obs}[1]{\hyperref[obs:#1]{Observation~\ref*{obs:#1}}} 
\newcommand{\Step}[1]{\hyperref[step:#1]{Step~\ref*{step:#1}}} 

\newcommand{\AH}{\mathbf{A}_{n,d}}

\newcommand{\pathtester}{\texttt{grid-path-tester}}

\newcommand{\RN}{\mathbb{R}}
\newcommand{\var}{\mathsf{var}}
\newcommand{\dist}{\mathtt{dist}}

\newcommand{\hM}{\widehat{M}}

\newcommand{\nofan}{{\sf Centrist}}

\newtheorem{theorem}{Theorem}[section]
\newtheorem{example}{Example}

\newtheorem{definition}[theorem]{Definition}
\newtheorem{lemma}[theorem]{Lemma}

\newtheorem{remark}[theorem]{Remark}

\newtheorem{fact}[theorem]{Fact}
\newtheorem{claim}[theorem]{Claim}

\renewcommand{\tilde}[1]{\widetilde{#1}}
\renewcommand{\Pr}{\mathbf{Pr}}

\newcommand{\map}{\mathsf{box}}
\newcommand{\bxx}{\mathsf{box}}
\newcommand{\fd}{f^{\mathsf{disc}}}
\begin{document}


\title{Domain Reduction for Monotonicity Testing: \\ A $o(d)$ Tester for Boolean Functions in $d$-Dimensions}


\author{
    Hadley Black\thanks{Department of Computer Science, University of California, Los Angeles. Email: \href{mailto:hablack@cs.ucla.edu}{\nolinkurl{hablack@cs.ucla.edu}}. Part of this work was done while the author was at University of California, Santa Cruz.}
    \and
    Deeparnab Chakrabarty\thanks{Department of Computer Science, Dartmouth College. Email: \href{mailto:deeparnab@dartmouth.edu}{\nolinkurl{deeparnab@dartmouth.edu}}. Supported by NSF CCF-1813053.}
    \and 
    C. Seshadhri\thanks{Department of Computer Science, University of California, Santa Cruz. Email: \href{mailto:sesh@ucsc.edu}{\nolinkurl{sesh@ucsc.edu}}. Supported by NSF TRIPODS CCF-1740850, CCF-1813165, and ARO Award W911NF191029.}
}
\date{}
\maketitle
\begin{abstract} 

We describe a $\tilde{O}(d^{5/6})$-query monotonicity
tester for Boolean functions $f \colon [n]^d \to \{0,1\}$ on the $n$-hypergrid. This is the first
$o(d)$ monotonicity tester with query complexity independent of $n$. 
Motivated by this independence of $n$, we initiate the study of monotonicity testing of 
measurable Boolean functions $f\colon\RR^d \to \{0,1\}$ over the continuous domain, where the distance
is measured with respect to a product distribution over $\RR^d$.
We give a $\tilde{O}(d^{5/6})$-query monotonicity tester for such functions. \smallskip

Our main technical result is a \emph{domain reduction theorem}
for monotonicity. For any function $f\colon[n]^d \to \{0,1\}$, 
let $\eps_f$ be its distance to monotonicity. 
Consider the restriction $\hat{f}$ of the function on a random $[k]^d$ sub-hypergrid of the original domain.
We show that for $k = \poly(d/\eps_f)$, the expected distance of the restriction is $\Exp[\eps_{\hat{f}}] = \Omega(\eps_f)$. 
Previously, such a result was only known for $d=1$ (Berman-Raskhodnikova-Yaroslavtsev, STOC 2014).
Our result for testing Boolean functions over $[n]^d$ then follows by applying the $d^{5/6}\cdot \poly(1/\eps,\log n, \log d)$-query hypergrid tester of Black-Chakrabarty-Seshadhri (SODA 2018).

To obtain the result for testing Boolean functions over $\RR^d$, 
we use standard measure theoretic tools to
reduce monotonicity testing of a measurable function $f$
to monotonicity testing of a discretized version of $f$ over a hypergrid domain $[N]^d$ for large, but finite, $N$ (that may depend on $f$).
The independence of $N$ in the hypergrid tester is crucial 
to getting the final tester over $\RR^d$.

\end{abstract}

\thispagestyle{empty}
\clearpage
\setcounter{page}{1}
\newpage

\section{Introduction}

Monotonicity testing is a fundamental problem in property testing. 
Let $(D, \prec)$ be a partially ordered set (poset) and let $R$ be a total order. 
A function $f\colon D \to R$ is monotone if $f(x) \leq f(y)$ whenever $x\prec y$. 
The hypercube, $\{0,1\}^d$ and the hypergrid $[n]^d$ have been the most studied posets in monotonicity testing, where
$\prec$ denotes the coordinate-wise partial ordering.
The Hamming distance between two functions $f$ and $g$ is $\dist(f,g) := \Pr_{x \sim D}[f(x) \neq g(x)]$ where $x$ is drawn uniformly from the domain. 
The distance of $f$ to monotonicity, denoted $\eps_f$, is its distance to the nearest monotone function.
That is, $\eps_f := \min_{g \in \cM} \dist(f,g)$, 
where $\cM$ is the set of all monotone functions. 
A monotonicity tester is a randomized algorithm that makes queries to $f$ and accepts with probability $\geq 2/3$ if the function is monotone, and rejects with probability $\ge 2/3$ if $\eps_f \geq \eps$, where $\eps\in (0,1)$ is an input parameter. 
The challenge is to determine the minimum query complexity of a monotonicity tester.


One of the earliest results in property testing is the $O(d/\eps)$-query ``edge-tester'' 
due to Goldreich et al.~\cite{GGLRS00} (see also~\cite{Ras99})
for testing monotonicity of Boolean functions over the hypercube, that is, $f\colon \{0,1\}^d \to \{0,1\}$.
In the last few years, considerable work~\cite{ChSe13-j,ChenST14,ChenDST15,KMS15,BeBl16,Chen17} has improved our understanding of Boolean monotonicity testing on the hypercube domain. In particular, Khot, Minzer, and Safra~\cite{KMS15} give an $\tilde{O}(\sqrt{d}/\eps^2)$ query\footnote{Throughout the paper $\widetilde{O}$ hides $\log(d/\eps)$ factors.}, non-adaptive tester, and Chen, Waingarten, and Xie~\cite{Chen17} show that any tester (even adaptive) must make
$\tilde{\Omega}(d^{1/3})$ queries. 
In contrast, for real-valued functions over the hypercube $f\colon\{0,1\}^d \to \mathbb{R}$, 
the complexity is known to be $\Theta(d/\eps)$~\cite{DGLRRS99,BBM11,ChSe13,ChSe14}, that is, linear in $d$. 


The problem of monotonicity testing Boolean functions $f\colon [n]^d \to \{0,1\}$ over hypergrids is not as well understood.
Dodis et al.~\cite{DGLRRS99} (with improvements by Berman, Raskhodnikova, and Yaroslavtsev~\cite{BeRaYa14}, henceforth BRY) give an $\tilde{O}(d/\eps)$-query tester.
The important feature to note is the {\em independence} of $n$.
Contrast this, again, with the real-valued case; monotonicity testing of functions $f\colon [n] \to \mathbb{R}$ requires $\Omega(\log n)$ queries~\cite{EKK+00,E04}. 
Recently, the authors~\cite{BlackCS18} describe an $\tilde{O}(d^{5/6}\log^{4/3} n ~\eps^{-4/3})$-query tester. Although the dependence on $d$ is sublinear, there is a dependence on $n$. The following question has remained open:
%
%
\emph{Is there a monotonicity tester for functions $f\colon [n]^d \to \{0,1\}$, whose query complexity is independent of $n$ and sublinear in $d$?}
%
\noindent One of the main outcomes of this work is an affirmative answer to this question.

\begin{theorem}\label{thm:main-informal-grid}
	There is a randomized algorithm that, given a parameter $\eps \in (0,1)$ and 
    query access to any Boolean function $f\colon [n]^d\to\{0,1\}$ defined over the hypergrid, makes $\tilde{O}(d^{5/6}\eps^{-4/3})$ non-adaptive queries to $f$ and (a) always accepts if $f$ is monotone, and (b) rejects with probability $> 2/3$ if $\eps_{f} > \eps$.
\end{theorem}


\noindent \textbf{Continuous Domains.} To the best of our knowledge,  monotonicity testing has so far been restricted to discrete domains. 
What can one say about monotonicity testing when the domain is $\RR^d$? Indeed, for functions whose range is $\RR$, the aforementioned lower bound of $\Omega(\log n)$ 
precludes any such tester (with finite query complexity) even in one dimension. 
On the other hand, the independence of $n$ in \Thm{main-informal-grid} (and indeed the results of Dodis et al.~\cite{DGLRRS99} and BRY~\cite{BeRaYa14}) suggests the possibility of a monotonicity tester
for Boolean functions $f\colon \RR^d\to\{0,1\}$. 
In this work, we spell out the natural definitions for monotonicity testing over $\RR^d$, and show that $o(d)$-testers
do exist when the distance is with respect to any product measure.




\begin{theorem} [\emph{Informal}, Formal version: \Thm{main-cont}]
	\label{thm:main-informal-cont} 
	There is a one-sided, non-adaptive $\tilde{O}(d^{5/6}\eps^{-4/3})$-query monotonicity tester 
	for {\em measurable} Boolean functions $f\colon\RN^d \to \{0,1\}$ with respect to arbitrary product measures\footnote{Each $\cD_i$ is described by a non-negative Lebesgue
    integrable function over $\RR$, whose integral over $\RR$ is $1$.}  $\cD = \prod_i \cD_i$.
\end{theorem}

To gain perspective, the reader may restrict attention to functions defined over the continuous cube $[0,1]^d$, and assume the uniform measure $\cD$ on this cube. This is the natural generalization of property testing on the domains $\{0,1\}^d$ and $[n]^d$ as described above.
The only restriction on the function we are testing is that the set of points where the function takes value $1$ (or $0$) must be (Lebesgue)-measurable. 
The distance between two functions $\dist(f,g) := \Pr_{x \sim \cD}[f(x)\neq g(x)]$ is the measure of the points at which they differ.
The distance to monotonicity of a function $f$ is $\inf_{g \in \cM} \dist(f,g)$ where $\cM$ is the set of all monotone functions. (In general, we use any measure to define distance.
For instance, we can test monotonicity of functions $f\colon\RR^d \to \{0,1\}$ over the Gaussian measure.) 

Note that the result of \Thm{main-informal-cont} holds
for all measurable functions, with no dependence on surface area or "complexity'' of $f$.
This can be contrasted with the recent result of De, Mossel, and Neeman~\cite{DeMN19}, who showed that Junta testing of Boolean functions $f\colon\RR^d \to \{0,1\}$ over the Gaussian measure requires some dependence on the surface area of $f$. 

Given the proof techniques for \Thm{main-informal-grid}, the proof of \Thm{main-informal-cont} follows
from standard measure theoretic methods. Nonetheless, we believe that there is a useful conceptual
message in \Thm{main-informal-cont}. It gives the natural ``limit" of monotonicity testing
for hypergrids $[n]^d$, as $n \rightarrow \infty$. This result also underscores the significance of getting testers independent of $n$ (for hypergrids), since it leads to testers for all measurable functions.

\subsection{Domain Reduction}

\paragraph{Discrete Hypergrid $[n]^d$.}
A natural approach to tackle Boolean monotonicity testing over the hypergrid is to try reducing it to Boolean monotonicity testing over the hypercube. 
For a function $f$ over $[n]^d$, consider 
the restriction $\hat{f}$ to a random hypercube in this hypergrid.
More precisely, for each dimension $i\in [d]$, sample two independent u.a.r. values $a_i < b_i$ in $[n]$ and 
let $\hat{f}$ be the restriction of $f$ on the hypercube formed by the Cartesian product $\prod_{i=1}^d \{a_i, b_i\}$.
If the expectation of $\eps_{\hat{f}}$ is $\Omega(\eps_f)$, then we obtain a hypergrid tester
by first reducing our domain to a random hypercube and then simply applying the best known monotonicity tester on the hypercube. 
However, we show that this does not work. In \Sec{lower}, we describe a function $f\colon [n]^d \to \{0,1\}$ 
such that $\eps_f = \Omega(1)$, but the restriction of $f$ on a random hypercube is monotone with probability $1 - \Theta(1/d)$ (see \Thm{lower_bound}). 

Nonetheless, one can consider the question of reducing the domain to a $[k]^d$ hypergrid, for
some parameter $k \ll n$, by sampling $k$ i.i.d. uniform elements of $[n]$ across each dimension. 
For $k$ {\em independent} of $n$, can we lower bound 
the expected distance of the function restricted to a random $[k]^d$ hypergrid?
BRY studied this question for the $d=1$ case (the line domain), and
prove that this is indeed possible~\cite{BeRaYa14}.
Our main technical result is a domain reduction theorem for all $d$,
by setting $k = \poly(d/\eps_f)$. 
That is, we show that \emph{if $k = \Theta((d/\eps_f)^7)$, then the expected distance to monotonicity of $f$ restricted to a random $[k]^d$ hypergrid is $\Omega(\eps_f)$.}

For a precise statement, let us fix a function $f\colon [n]^d \to \{0,1\}$.
Construct $d$ random (multi-) sets $T_1,\ldots,T_d \subseteq [n]$, each formed by taking $k$ i.i.d. uniform samples from $[n]$. Define $\bT := T_1 \times \cdots \times T_d$ and let $f_{\bT}$ denote $f$ restricted to $\bT$.
(We treat duplicate elements of a multi-set as being distinct copies of that element, which are then treated as immediate neighbors in the total order.) 

\begin{theorem} [Domain Reduction Theorem for Hypergrids] \label{thm:dir_domain_reduction} Let $f\colon [n]^d \to \{0,1\}$ be any function 
	and let $k \in \mathbb{Z}^+$ be a positive integer. If $\bT = T_1 \times \cdots \times T_d$ is a randomly chosen sub-grid, where for each $i \in [d]$, $T_i$ is a (multi)-set formed by taking $k$ i.i.d. samples from the uniform distribution on $[n]$, then
	
	\[\Expect_{\bT}\left[\eps_{f_{\bT}}\right] \geq \eps_f - \frac{C \cdot d}{k^{1/7}}\]
	\noindent
	where $C > 0$ is a universal constant. In particular, if $k \geq \left(\frac{2Cd}{\eps_f}\right)^7$, then $\Expect_{\bT}\left[\eps_{f_{\bT}}\right] \geq \eps_f/2$.
\end{theorem}

%
%

\noindent
The construction in \Sec{lower} shows that such a theorem is impossible for $k = o(\sqrt{d})$,
and thus, domain reduction requires $k$ and $d$ to be polynomially related. We leave figuring out the best dependence on $k$ and $d$ as an open question. For the $d=1$ case,
BRY give a much better lower bound of $\eps_f - 5\sqrt{\eps_f/k}$ (Theorem 3.1 of~\cite{BeRaYa14}). 


Given \Thm{dir_domain_reduction}, one can 
sample a random $[k]^d$ hypergrid denoted $\bT$ and apply the tester in~\cite{BlackCS18}
on $f_{\bT}$. The final query complexity is $\tilde{O}(d^{5/6})\cdot \poly\log k$. 
Setting $k = \poly(d/\eps)$, one gets a purely sublinear-in-$d$ tester (see \Sec{tester} for a formal proof). An obvious question is whether the dependence on $d$ can be brought down to $\sqrt{d}$ as in the hypercube case. If one could design a $\sqrt{d} \cdot \poly \log n$ query monotonicity tester for the domain $[n]^d$, 
then \Thm{dir_domain_reduction} can be used as a black box to achieve an $\tilde{O}(\sqrt{d})$ monotonicity tester. Note that because the dependence of~\cite{BlackCS18} is $\poly\log k$, and in light of the fact that $k = \poly(d)$ is needed for domain reduction to hold (\Thm{lower_bound}), any
improvement to \Thm{dir_domain_reduction} would only give a constant factor improvement to the query complexity of the overall tester. \\



\noindent \textbf{Continuous Domains.} The independence of $n$ in \Thm{dir_domain_reduction} suggests the possibility of a domain reduction result for Boolean functions defined over $\RN^d$.
We show that this is indeed true if  $f\colon \RN^d \to \{0,1\}$ is measurable (formal definitions in~\Sec{discrete}) and defined with respect to a (Lebesgue integrable) product distribution.

\begin{theorem} [Domain Reduction Theorem for $\RN^d$] \label{thm:cont_domain_reduction} Let $f\colon\RN^d \to \{0,1\}$ be any measurable function and let $k\in \mathbb{Z}^+$ be a positive integer. Let $\cD = \prod_{i=1}^d \cD_i$ be a (Lebesgue integrable) product distribution such that the distance to monotonicity of $f$ w.r.t. $\cD$ is $\eps_f$. If $\bT = T_1 \times \cdots \times T_d$ is a randomly chosen hypergrid, where for each $i \in [d]$, $T_i \subset \RN$ is formed by taking $k$ i.i.d. samples from $\cD_i$, then
$\Expect_{\bT}\left[\eps_{f_{\bT}}\right] \geq \eps_f - \frac{C \cdot d}{k^{1/7}}$,
where $C > 0$ is a universal constant. In particular, if $k \geq \left(\frac{2Cd}{\eps_f}\right)^7$, then $\Expect_{\bT}\left[\eps_{f_{\bT}}\right] \geq \eps_f/2$.
\end{theorem}

The above theorem essentially reduces the continuous domain to a discrete hypergrid $[k]^d$ where $k$ is at most some polynomial of the dimension $d$.
At this point, our result from~\cite{BlackCS18} implies \Thm{main-informal-cont}; a formal proof is given in \Sec{tester}.

The main ingredient in the proof of \Thm{cont_domain_reduction} is a discretization lemma (\Lem{good-approx}). Using standard measure theory, one can show that 
for any measurable Boolean function over $\RN^d$ and any $\delta > 0$, there exists a large enough natural number $N = N(f,\delta)$ 
with the following property. The domain $\RR^d$ can be divided into an $N^d$ sized
$d$-dimensional grid, such that in at least a $(1-\delta)$-fraction of grid boxes,
the function $f$ has the same value. (In some sense, this is what it means for $f$ to be measurable.)
Ignoring the $\delta$-fraction of "mixed" boxes, the function $f$ can be thought
of as a discrete function on $[N]^d$.

%
%
%

\emph{The only guarantee on $N$ is that it is finite}; as it depends on $f$, $N$ could be extremely large compared to $d$. This is where \Thm{dir_domain_reduction} shows its power.
The sampling parameter $k$ is independent of $N$, and this establishes \Thm{cont_domain_reduction}. We give a detailed proof in \Sec{proof_cont_domain_reduction}.

We remark here that given the discretization lemma (\Lem{good-approx}), one can also apply the techniques of Dodis et al.~\cite{DGLRRS99} and BRY~\cite{BeRaYa14} to get an $\tilde{O}(d/\eps)$-query tester. 
However, as we mentioned before, we are unaware of an explicit study of monotonicity testing over the continuous domain.

\subsection{Related Work} \label{sec:prev}

Monotonicity testing has been extensively studied in the past two decades~\cite{EKK+00,GGLRS00,DGLRRS99,LR01,FLNRRS02,HK03,AC04,HK04,ACCL04,E04,SS08,Bha08,BCG+10,FR,BBM11,RRSW11,BGJ+12,ChSe13,ChSe13-j,ChenST14,BeRaYa14,BlRY14,ChenDST15,ChDi+15,KMS15,BeBl16,Chen17,BlackCS18}.

We give a short summary of Boolean monotonicity testing over the hypercube. The problem was introduced
by Goldreich et al.~\cite{GGLRS00} (also refer to Raskhodnikova's thesis~\cite{Ras99}), who describe an $O(d/\eps)$-query tester.
The first improvement over that bound was the $\tilde{O}(d^{7/8})$ tester due to Chakrabarty and Seshadhri~\cite{ChSe13-j}, achieved via
a directed analogue of Margulis' isoperimetric theorem. Chen-Servedio-Tan~\cite{ChenST14} improved the analysis to get an $\tilde{O}(d^{5/6})$ bound.
A breakthrough result of Khot-Minzer-Safra~\cite{KMS15} gives an $\tilde{O}(\sqrt{d})$ tester. All of these testers are non-adaptive and one-sided.
Fischer et al.~\cite{FLNRRS02} prove a (nearly) matching lower bound of $\Omega(\sqrt{d})$ for this case.
The first polynomial two-sided lower bound was given in Chen-Servedio-Tan~\cite{ChenST14} and was subsequently improved to $\Omega(d^{1/2-\delta})$
in Chen et al.~\cite{ChenDST15}. The first polynomial lower bound of $\tilde{\Omega}(d^{1/4})$ for adaptive testers was given
in Belovs-Blais~\cite{BeBl16} and has since been improved to $\tilde{\Omega}(d^{1/3})$ by Chen-Waingarten-Xie~\cite{Chen17}.

For Boolean monotonicity testing over general hypergrids, Dodis et al.~\cite{DGLRRS99} give a non-adaptive, one-sided $O((d/\eps)\log^2(d/\eps))$-query 
tester. This was improved to $O((d/\eps)\log(d/\eps))$ by Berman, Raskhodnikova and Yaroslavtsev~\cite{BeRaYa14}.
This paper also proves an $\Omega(\log(1/\eps))$ separation between adaptive and non-adaptive monotonicity testers for $f\colon [n]^2 \to \{0,1\}$ by demonstrating an $O(1/\eps)$ adaptive tester (for any constant $d$), and an $\Omega(\log(1/\eps)/\eps)$ lower bound
for non-adaptive monotonicity testers. Previous work by the authors~\cite{BlackCS18} gives a monotonicity tester
with query complexity $\tilde{O}(d^{5/6}\log^{4/3} n)$ via directed isoperimetric inequalities for augmented hypergrids.

\subsection{Further Remarks}

\noindent \textbf{Implication for Other Notions of Distance:} 
Berman, Raskhodnikova, and Yaroslavtsev~\cite{BeRaYa14} introduce
the notion of $L_p$-testing, where $f\colon [n]^d \to [0,1]$ and the distance between functions is measured
in terms of $L_p$-norms~\cite{BeRaYa14}. They prove (Lemma 2.2 + Fact 1.1,~\cite{BeRaYa14}) that $L_p$-monotonicity testing can be reduced
to (non-adaptive, one-sided) Boolean monotonicity testing. Thus, \Thm{main-informal-grid} implies an $L_p$-monotonicity tester for functions $f\colon [n]^d \to [0,1]$ which makes $o(d)$ queries
. This improves upon Theorem 1.3 of~\cite{BeRaYa14}. 

We also believe our main theorem~\Thm{main-informal-grid} can be used to {\em estimate} the distance-to-monotonicity for functions $f\colon [n]^d \to \{0,1\}$ in time independent of $n$.
The works of~\cite{BeRaYa14,PRR04} also relate distance estimation for Boolean functions and {\em tolerant} testing over $L_p$-distances, 
and our results should have implications for this. Finally, generalizing $L_p$-testing to the continuous domain should be possible.
We leave all these interesting directions as future work.
\smallskip

\noindent \textbf{Domain Reduction for Variance:} Recent works~\cite{ChSe13-j,KMS15,BlackCS18} have shown that 
certain isoperimetric theorems for the undirected hypercube have \emph{directed analogues} 
where the variance is replaced by the distance to monotonicity. 
Interestingly, for the case of domain reduction, the variance and distance to monotonicity behave differently. 
While domain reduction for the distance to monotonicity requires $k \geq \Omega(\sqrt{d})$ (\Thm{lower_bound}), 
we show that the expected variance of a restriction of $f$ to a random hypercube ($k=2$) is at least half
the variance of $f$ 
(see \Thm{undir_domain_reduction}). This statement may be of independent interest.
We were unable to find a reference to such a statement and provide a proof in \Sec{undirected}.

\section{Proving the Domain Reduction \Thm{dir_domain_reduction}: Overview} \label{sec:overview}
The theorem is a direct corollary of
the following lemma, applied to each dimension.

\begin{lemma} [Domain Reduction Lemma] \label{lem:dom_red_main} Let $f\colon [n] \times \left(\prod_{i=2}^d [n_i]\right) \to \{0,1\}$ be any function over a rectangular hypergrid for some $n, n_2,\ldots,n_d \in \mathbb{Z}^+$ and let $k \in \mathbb{Z}^+$. Choose $T$ to be a (multi-) set formed by taking $k$ i.i.d. samples from the uniform distribution on $[n]$ and let $f_T$ denote $f$ restricted to $T \times \left(\prod_{i=2}^d [n_i]\right)$. Then $\Expect_T\left[\eps_f - \eps_{f_T}\right] \leq \frac{C}{k^{1/7}}$ where $C > 0$ is a universal constant. 
\end{lemma}
This lemma is the heart of our results, and in this section we give an overview of its proof.
Let us start with the simple case of $d=1$ (the line).
Monotonicity testers for the line immediately imply domain reduction for $d=1$~\cite{DGLRRS99,BeRaYa14}.
A u.a.r. sample of $\tilde{O}(1/\eps_f)$ points in $[n]$ contains a monotonicity violation with large probability ($>9/10$, say),
and thus the restriction of $f$ to this sample has distance $\widetilde{\Omega}(\eps_f)$. 
However, $\Omega(\eps_f)$ is weak for what we need since, 
%
even if one could generalize this argument to the setting of \Lem{dom_red_main},
we would need to apply it $d$ times to get the full domain reduction (\Thm{dir_domain_reduction}). This would imply a final lower bound of $\eps_f/C^d$, for some constant $C$, which has little value towards proving a sublinear-in-$d$ query tester.   

Fortunately, quantitatively stronger domain reduction exists
for the line. BRY (\cite{BeRaYa14}, Theorem 3.1) proves that if one samples $\Theta(s^2/\eps_f)$ points, then the expected distance of the restricted function
is at least $\eps_f(1 - 1/s)$.
Numerically speaking, this is encouraging news, since we could try to set $s = \Theta(d)$ and
iterate this argument $d$ times (over each dimension). Of course, this result for the line alone is not enough to deal with the structure of general hypergrids, but forms a good sanity check.

Consider the general case of \Lem{dom_red_main}. For brevity, we let $D := [n] \times \left(\prod_{i=2}^d [n_i]\right)$ and $D_T := T \times \left(\prod_{i=2}^d [n_i]\right)$ denote the original and reduced domains, respectively. Note that $|D_T| = \frac{k}{n}|D|$. 

The standard handle on the distance to monotonicity is the violation graph of $f$, arguably first formalized by Fischer et al.~\cite{FLNRRS02}.
The graph has vertex set $D$ and an edge $(x,y)$ iff $x\prec y$ and $f(x) = 1, f(y) = 0$.
A theorem of~\cite{FLNRRS02} states that any maximum cardinality matching $M$ in the violation graph satisfies $|M| = \eps_f|D|$. Fix such a matching $M$.
For a fixed sample $T$, we let $M_T$ denote a maximum cardinality matching in the violation graph of $f_T$. To argue about $\eps_{f_T}$, we want 
to give a lower bound on the expected size $|M_T|$. To do so, we give a lower bound the expected number of endpoints of $M$ that can still be matched (simultaneously) in the violation graph of $f_T$. 
\smallskip


We use the following standard notions of lines and slices in $D$, with respect to the first dimension. Refer to \Fig{stack-example} and \Fig{2d} for visual examples in two dimensions. In these examples the rows represent the lines while the columns represent the slices.
Below, for $x\in D$, the vector $x_{-1}$ is used to denote $(x_2,x_3,\ldots, x_d)$.
\begin{itemize} [noitemsep]
    \item (Lines in $D$) $\cL := \left\{\ell_z : z \in \prod_{i=2}^{d}[n_i]\right\}$ where $\ell_z := \left\{x \in D : x_{-1} = z\right\}$.
    \item (Slices in $D$) $\cS := \left\{S_i : i \in [n]\right\}$ where $S_i := \left\{x \in D : x_1 = i\right\}$.
\end{itemize}




We partition $M$ into a collection of "local" matchings for each line:

\begin{itemize}
    \item (Line Decomposition of $M$) For each $\ell \in \cL$: $M^{(\ell)} := \{(x,y) \in M : x \in \ell\}$.
\end{itemize}

We find a large matching in the violation graph of $f_T$ by doing a line-by-line analysis. In particular, for each line $\ell \in \cL$, we define the following matching $M^{(\ell)}_T$ in the violation graph of $f_T$.

\begin{itemize}
\item (The matching $M^{(\ell)}_T$) For each $\ell\in \cL$, consider the collection of all maximum cardinality violation matchings w.r.t. $f_T$ on the set of vertices that (a) are matched by $M^{(\ell)}$, and (b) lie in some slice $S_i$ where $i \in T$. We let $M^{(\ell)}_T$ denote any such fixed matching.
\end{itemize}


We stress that $M^{(\ell)}_T$ is {\em not} a subset of $M^{(\ell)}$, but the \emph{endpoints} of the pairs in $M^{(\ell)}_T$ are a subset of the \emph{endpoints} of the pairs in $M^{(\ell)}$.
Thus, by the above definition, the union $M_T := \cup_{\ell \in \cL} M_{T}^{(\ell)}$ is a valid matching in the violation graph of $f_T$ since $M^{(\ell)}$ and $M^{(\ell')}$ have disjoint endpoints for all $\ell \neq \ell' \in \cL$. We will lower bound the size of this matching, $|M_T|$, by giving a lower bound on $|M_T^{(\ell)}|$ for each line $\ell$.

Fix some $\ell \in \cL$. By definition, the lower-endpoints of $M^{(\ell)}$ all lie on $\ell$, and thus are all comparable. Let $M^{(\ell)} = \{(x_1,y_1),\ldots,(x_m,y_m)\}$ where $x_1 \prec \cdots \prec x_m$ and observe that, for any $j \in [m]$, $x_1,\ldots,x_j \prec y_j,\ldots,y_m$. 
Since the function is Boolean, every $x \in \{x_1,\ldots,x_j\}$ forms a violation to monotonicity with every $y \in \{y_j,\ldots,y_m\}$, 
and therefore these vertices can be matched in $M^{(\ell)}_T$, if their $1$-coordinates are sampled by $T$. 

Since all the $x_i$'s lie on the same line $\ell$, their $1$-coordinates are distinct. Suppose that the $1$-coordinates of all the $y_i$'s were also distinct and distinct from those of the $x_i$'s too. 
Under this assumption we can proceed with our analysis as if all the $x_i$'s \emph{and} $y_i$'s lie on $\ell$, and the analysis becomes identical to the one-dimensional case.
We could thus apply Theorem 3.1 of \cite{BeRaYa14} to each $\ell \in \cL$ to prove \Lem{dom_red_main}. 
However, the assumption that the $y_i$'s have distinct $1$-coordinates is far from the truth. 
As we explain below, there are examples where all the $y_i$'s have the same $1$-coordinate,
thereby lying in the same slice $S_a$ (for some $a \in [n]$). In this case, with probability $(1-k/n)$ we would have the size of $M^{(\ell)}_T$ be $0$ (if $a\notin T$), 
implying that $\Expect_T\left[|M^{(\ell)}_T|\right]$ could be as small as $(k/n)^2\cdot |M^{(\ell)}|$. 
Thus, if there existed a function $f$ such that a ``collision of $y$'s $1$-coordinates'' could not be avoided for a large number of lines, then this would preclude such a line-by-line approach to proving \Lem{dom_red_main}.
Unfortunately, there are examples of violation matchings where this happens. Consider \Ex{2d}, and the left part
of \Fig{2d}, shown at the end of this section. For the lowest line, all the corresponding $y$'s in $M^{(\ell)}$ have the
same $1$-coordinate. 

Our main insight is that for any $f$, there {\em always} exists a violation matching $M$ where the problem
above does not arise too often. This motivates the key definition of {\em stacks}; the stacks are what determine the ``shape'' of a matching.
Formally, for any $\ell \in \cL$ and $S \in \cS$, the $(\ell,S)$-stack is the set of pairs $(x,y)\in M$, where $x \in \ell$ and $y \in S$.

\begin{itemize} [noitemsep]
    \item (Stacks) $M^{(\ell,S)} := \{(x,y) \in M^{(\ell)} : y \in S\} = \{(x,y) \in M : x \in \ell, y \in S\}$.
\end{itemize}


We call $|M^{(\ell,S)}|$ the ``{\em size} of the stack $(\ell,S)$''. To summarize the above discussion, small stacks are good news while big stacks are bad news. This is formalized in \Lem{line_sampling}.


%
%
%
%
%

If there is a maximum cardinality matching $M$ in the violation graph of $f$ such that all stacks 
have size at most $1$, then the one-dimensional domain reduction can be directly applied. 
Unfortunately, this is not possible. We give an example in \Fig{stack-example} of a function where stacks of size at least $2$ are unavoidable\footnote{Interestingly, we don't know of a function where stacks of size strictly larger than $2$ can't be avoided. In fact, we can prove that for the grid (the $d=2$ case) one can always find a maximum cardinality violation matching $M$ where $|M^{(\ell,S)}| \leq 3$ for all $(\ell,S)$. The proof is cumbersome and so we exclude it since it is not relevant to our main result.}.
One reason for this difficulty may be that there can be various maximum cardinality matchings in the violation graph 
that have vastly different stack sizes (shapes); 
again consider~\Ex{2d}. Nevertheless, we prove that 
there is a matching $M$ such that for every positive integer $\lambda$, the 
total number of pairs belonging to
stacks of size at least $\lambda$ is at most $|D|/\poly(\lambda)$. 

\begin{figure}[h]
    \begin{center}
    \includegraphics[clip, scale=0.55]{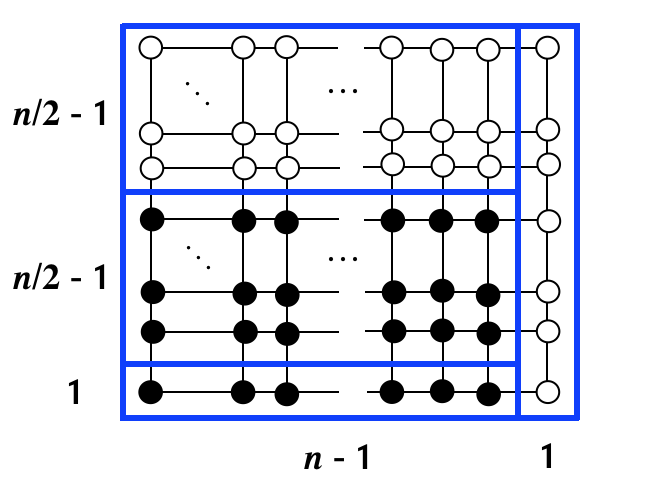}
    \end{center}
    \caption{An example of a function $f \colon [n] \times [n-1] \to \{0,1\}$ where stacks of size $\geq 2$ are unavoidable. Black (white, resp.) circles represent vertices where $f=1$ ($f=0$, resp.). First observe that there exists a perfect violation matching as follows: perfectly match the two blocks of size $(n-1)(n/2-1)$ and then perfectly match the bottom line of $1$'s to the right-most slice of $0$'s. Thus, any maximum cardinality violation matching, $M$, will match all of the $(n-1)$ $0$'s in the right-most slice. There are only $n/2$ lines containing $1$'s and so by the pigeonhole principle $M$ contains at least $n/2-1$ pairs belonging to stacks of size $\geq 2$.}
    \label{fig:stack-example}
\end{figure}

\begin{lemma} [Stack Bound] \label{lem:stack_bound} There exists a maximum cardinality matching $M$ in the violation graph of $f$ such that for every $\lambda \in \mathbb{Z}^+$, $M$ satisfies $\sum_{(\ell,S) : |M^{(\ell,S)}| \geq \lambda} |M^{(\ell,S)}| \leq \frac{5}{\sqrt{\lambda}} \cdot |D|$. \end{lemma}

The main creativity to prove this lemma lies in the choice of $M$. Given a matching, we define the vector $\Lambda(M)$ that enumerates all the stack sizes in non-decreasing order.
We show that the maximum cardinality matching $M$ with the lexicographically largest $\Lambda(M)$
serves our purpose. That is, we choose $M$ that maximizes the minimum stack size, and then subject to this maximizes the second minimum, and so on. 
It may seem counter-intuitive that we want a matching with small stack sizes, and yet our potential function maximize the minimum. 
The intuitive explanation is that the sum of the stack sizes is $|M|$, which is fixed, and so in a sense maximizing the minimum also balances out the $\Lambda(M)$ vector.
The proof uses a matching rewiring argument to show that
any large stack must be "adjacent" to many moderate size stacks. If two stacks are appropriately
"aligned", one could change the matching to move points from one stack to the other. 
Large stacks cannot be aligned with small stacks, since one could rewire the matching to increase
the potential.
But since the function is Boolean one can show that there are many opportunities for rewiring the violation matching. 
Thus, there isn't enough "room" for many large stacks.
We then apply some technical charging arguments to bound the total number of points in large stacks.
The full proof is given in \Sec{stack_bound}.



With the stack bound in hand, we need to generalize the one-dimensional argument of BRY (Theorem 3.1 \cite{BeRaYa14}) to account for bounded
stack sizes. Then, we bound $|M^{(\ell)}_T|$
for all $\ell$, and get the final lower bound on the distance $\eps_{f_T}$.

\begin{lemma} [Line Sampling] \label{lem:line_sampling} Suppose that $M$ is a matching in the violation graph of $f$,  such that for some $\lambda \in \mathbb{Z}^+$, $|M^{(\ell,S)}| \leq \lambda$ for all $\ell \in \cL$ and $S \in \cS$. Then, for any $\ell \in \cL$,
    
\[\Expect_{T}\left[|M^{(\ell)}_T|\right] \geq \frac{k}{n} \cdot |M^{(\ell)}| - 3 \lambda \sqrt{k \ln k} \text{.}\]
\end{lemma}


The proof is a fairly straightforward generalization of the arguments in~\cite{BeRaYa14} for the $\lambda=1$ case.
The idea is to control the size of the maximum cardinality matching $M^{(\ell)}_T$ by analyzing the discrepancy of a random subsequence of a sequence of $1$s and $0$s. 
For the sake of simplicity, 
we give a proof that achieves a weaker dependence on $\eps_f$ than in~\cite{BeRaYa14}. 
Our proof of \Lem{line_sampling} is given in \Sec{discrepancy}. 
We note that BRY give a stronger lower bound (without the $\sqrt{\ln k}$) and also bound the variance for the $\lambda = 1$ case. A more careful generalization of BRY which removes the $\sqrt{\ln k}$ would yield an improved loss of $C/k^{1/6}$ instead of $C/k^{1/7}$ in \Lem{dom_red_main}, but we prefer to give the simpler $C/k^{1/7}$ exposition for the purpose of ease of reading.
\begin{example}[A Two Dimensional Example]\label{ex:2d}
{\em
    Consider the anti-majority function on two dimensions. More precisely, let $f\colon [n]^2 \to \{0,1\}$ be defined as $f(x,y) = 1$ if $x + y \leq n$, and $f(x,y) = 0$ otherwise.
We describe two maximum cardinality matchings with vastly different stack sizes. The first matching $R$ matches a point $(x,y)$ with $x+y\leq n$ to the point $(n-y+1,n-x+1)$.
For an illustration, see the left matching in \Fig{2d} for the case $n=5$. Observe that whenever $x+y\leq n$, we have $(n-y+1) + (n-x+1) > n$.
The second matching $B$ matches a point $(x,y)$ with $x+y \leq n$ to the point $(x+y, n-x+1)$. Again, observe that $(x+y) + (n-x+1) > n$.
For an illustration, see the right blue matching in \Fig{2d} for the case $n=5$. Note that the stack sizes for the matching $R$ are large; in particular, they are $n-1,n-2,\ldots,2,1$ for $n-1$ stacks and $0$ for the rest.
On the other hand, any stack in $B$ is of size $\leq 1$. 

\begin{figure}[h]
    \begin{center}
    \includegraphics[trim= 100 280 0 80, clip, scale=0.7]{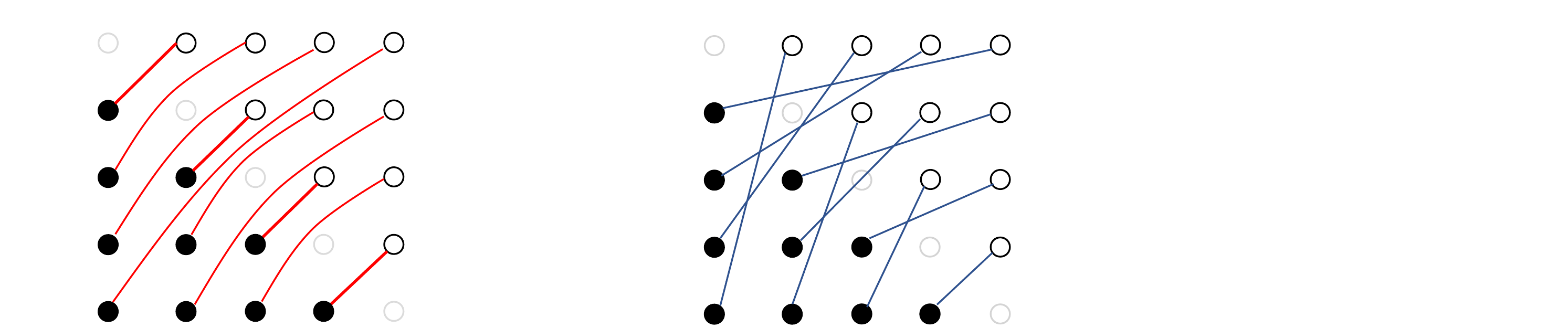}
    \end{center}
    \caption{Accompanying illustration for \Ex{2d} showing two different maximum cardinality violation matchings for the anti-majority function $f\colon [5]^2 \to \{0,1\}$ which have very different stack sizes. Black (white, resp.) circles represent vertices where $f = 1$ ($f=0$, resp.) and connecting lines represent pairs of the matching. Observe that for the left matching, the bottom line and the right-most slice form a stack of size $4$ while the right matching has stack sizes all $\leq 1$. }
    \label{fig:2d}
\end{figure}

%
%
%
%
%
%
%
%
}
\end{example}

\section{Domain Reduction: Proof of \Lem{dom_red_main}} \label{sec:dom_red_main}

In this section, we use \Lem{stack_bound} and \Lem{line_sampling} to prove \Lem{dom_red_main}.
Recall that $D := [n] \times \left(\prod_{i=2}^d [n_i]\right)$ and $D_T := T \times \left(\prod_{i=2}^d [n_i]\right)$ denote the original and reduced domains, respectively. Note that $|D_T| = \frac{k}{n} |D|$. Let $M$ be the matching given by \Lem{stack_bound} and consider $\lambda = \left\lceil25k^{2/7}\right\rceil$. Clearly, $\lambda \in [25k^{2/7}, 26k^{2/7}]$. 

\noindent
Thus, by \Lem{stack_bound}, we have 
%
    $\left|\bigcup_{(\ell,S) : |M^{(\ell,S)}| \geq 26k^{2/7}} M^{(\ell,S)}\right| \leq \frac{5}{\sqrt{25k^{2/7}}} \cdot |D| = \frac{|D|}{k^{1/7}}$.
Let 
$$ \hM := M \setminus \left( \bigcup_{(\ell,S): |M^{(\ell,S)}| \geq 26k^{2/7}} M^{(\ell,S)} \right)$$
denote the set of pairs in $M$ which do not belong to stacks larger than $26k^{2/7}$; we therefore have


\begin{align} \label{eq:stack_loss}
    \sum_{\ell \in \cL} |\hM^{(\ell)}| = |\hM| \geq |M| - \frac{|D|}{k^{1/7}} \text{.}
\end{align}

In this proof, our goal is to \emph{construct} a matching $M_T$ in the violation graph of $f_T$ whose cardinality is sufficiently large. We measure $\Expect_{T}\left[|M_T|\right]$ by summing over all lines in $\cL$ and applying \Lem{line_sampling} to each. Notice that $\hM$ is a matching in the violation graph of $f$ which satisfies $|\hM^{(\ell,S)}| \leq 26 k^{2/7}$ for all $\ell \in \cL$ and $S \in \cS$. Thus by \Lem{line_sampling}, for every $\ell \in \cL$,

\begin{align} \label{eq:line_size}
    \Expect_{T}\left[|M^{(\ell)}_T|\right] \geq \frac{k}{n} \cdot |\hM^{(\ell)}| - 3 \cdot (26 k^{2/7}) \cdot \sqrt{k \ln k} \geq \frac{k}{n} \cdot |\hM^{(\ell)}| - 78 k^{5/6}
\end{align}

\noindent where we have used $\sqrt{\ln k} < k^{1/3-2/7}$. Now, using \Eqn{stack_loss} and \Eqn{line_size}, we can calculate $\Expect_{T}\left[|M_T|\right]$. We use the fact that $\{\hM^{(\ell)}\}_{\ell \in \cL}$ is a partition of $\hM$, apply linearity of expectation and use \Lem{line_sampling} to measure $\Expect_{T}[|M^{(\ell)}_T|]$ for each $\ell$. Also note that the number of lines is $|\cL| = |D|/n$.

\begin{align} \label{eq:crossing-size}
    \Expect_{T}\left[\left|M_T\right|\right] &= \Expect_{T}\left[\sum_{\ell \in \cL} |M^{(\ell)}_T|\right] = \sum_{\ell \in \cL} \Expect_{T}\left[|M^{(\ell)}_T|\right] \geq \sum_{\ell \in \cL} \left(\frac{k}{n} \cdot |\hM^{(\ell)}| - 78 k^{5/6}\right) \text{ (by \Eqn{line_size})}\nonumber \\
    &= \left(\frac{k}{n} \cdot \sum_{\ell \in \cL} |\hM^{(\ell)}|\right) - \left(78 k^{5/6} \cdot \frac{|D|}{n}\right) \geq \frac{k}{n} \cdot \left(|M| - \frac{|D|}{k^{1/7}}\right) - \left(78 k^{5/6} \cdot \frac{|D|}{n}\right) \text{ (by \Eqn{stack_loss})}\nonumber \\
    &= \frac{k}{n} \cdot \left(|M| - \frac{|D|}{k^{1/7}} - \frac{78 |D|}{k^{1/6}} \right) \geq \frac{k}{n}\cdot\left(|M| - \frac{C \cdot |D|}{k^{1/7}}\right)
\end{align}



\noindent for a constant $C > 0$, since $\frac{1}{k^{1/7}}$ dominates $\frac{1}{k^{1/6}}$. \Eqn{crossing-size} gives the expected cardinality of our matching after sampling. To recover the distance to monotonicity we simply normalize by the size of the domain. Dividing by $|D_T| = \frac{k}{n} |D|$, we get $\Expect_{T}\left[\eps_{f_T}\right] \geq \frac{|M|}{|D|} - \frac{C}{k^{1/7}} = \eps_{f} - \frac{C}{k^{1/7}}$. 
This completes the proof of \Lem{dom_red_main}. \pushQED{\qed} \qedhere \popQED

\section{Stack Bound: Proof of \Lem{stack_bound}} \label{sec:stack_bound}

We are given a Boolean function $f\colon D \to \{0,1\}$ where $D = [n] \times \left(\prod_{i=2}^d [n_i]\right)$ is a rectangular hypergrid for some $n,n_2,\ldots,n_d \in \mathbb{Z}^+$. \Lem{stack_bound} asserts there is a maximum cardinality matching $M$ such that $\sum_{(\ell,S) : |M^{(\ell,S)}| \geq \lambda} |M^{(\ell,S)}| \leq \frac{5}{\sqrt{\lambda}}\cdot |D|$ for all $\lambda \in \mathbb{Z}^+$. \smallskip

Given a matching $M$, we consider the vector (or technically, the list) $\Lambda(M)$ indexed by stacks $(\ell,S)$ with $\Lambda_{\ell,S} := |M^{(\ell,S)}|$, and 
list these in {\em non-decreasing} order.
Consider the maximum cardinality matching $M$ in the violation graph of $f$ which has the lexicographically largest $\Lambda(M)$. That is, the minimum entry of $\Lambda(M)$ is maximized, and subject to that the second-minimum is maximized and so on. We fix this matching $M$ and claim that it  
satisfies $\sum_{(\ell,S) : |M^{(\ell,S)}| \geq \lambda} |M^{(\ell,S)}| \leq \frac{5}{\sqrt{\lambda}}\cdot |D|$ for all $\lambda \in \mathbb{Z}^+$. Note that the inequality is trivial for $\lambda \leq 100$, since $M$ itself is of size at most $\eps_f|D| \leq \frac{1}{2}|D|$. Thus, in what follows we prove that the inequality is true for an arbitrary, fixed $\lambda > 100$. We first introduce the following notation.

\begin{itemize} [noitemsep]
    \item (Low Stacks) $L := \{(\ell,S) \in \cL \times \cS: |M^{(\ell,S)}| \leq \lambda - 2\}$.
    \item (High Stacks) $H := \{(\ell,S) \in \cL \times \cS: |M^{(\ell,S)}| \geq \lambda\}$.
\end{itemize}

Let $V(H)$ denote the set of vertices matched by $\bigcup_{(\ell,S) \in H} M^{(\ell,S)}$. 
Let $B$ (for blue) be the set of points in $V(H)$ with function value $0$, and $R$ (for red) be the set of points in $V(H)$ with function value $1$.
$M$ induces a perfect matching between $B$ and $R$, and we wish to prove $|B| = |R| \leq  \frac{5}{\sqrt{\lambda}} \cdot |D|$. 
Indeed, define $\delta$ to be such that $|B| = \delta|D|$. In the remainder of the proof, we will show that $\delta \leq \frac{5}{\sqrt{\lambda}}$. \smallskip


We make a simple observation that for any fixed line $\ell$, there cannot be too many non-low stacks $(\ell, S)$.

\begin{claim} \label{clm:trivial_bound} For any line $\ell$, the number of non-low stacks $\ell$ participates in is at most $\frac{n}{\lambda - 1}$. \end{claim}

\begin{proof} Fix any line $\ell$ and consider the set $\bigcup_{S : (\ell,S) \notin L} \left\{x_1 : \exists(x,y) \in M^{(\ell,S)}\right\}$.
That is, the set of 1-coordinates that are used by 
some non-low stack involving $\ell$. 
The size of this set can't be bigger than the length of $\ell$, which is $n$.
	Furthermore, each non-low stack contributes at least $\lambda - 1$ {\em unique} entries to this set. The uniqueness follows since the union $\bigcup_{S : (\ell,S) \notin L} M^{(\ell,S)}$ is a matching.
%
\end{proof}

We show that if the number of blue points $|B|$ is large ($> 5|D|/\sqrt{\lambda}$), then we will find a line participating in more than $n/(\lambda - 1)$ non-low stacks. To do so, we need to ``find'' these non-low stacks. 
%
We need some more notation to proceed. For a vertex $z$, we let $\ell_z$ ($S_z$, resp.) denote the unique line (slice, resp.) containing $z$. For each blue point $y \in B$, we define the following interval 
	\[
		\cI_y := \left\{z \in \ell_y : z_1 \in [x_1,y_1]\right\} \subseteq \ell_y ~~\textrm{where}~~ (x,y)\in M \text{.}
	\] 
Note that $\cI_y$ is the interval of $\ell_y$ whose endpoints are given by the projection of $(x,y)$ onto $\ell_y$. Armed with this notation, we can find our non-low stacks. Our next claim, which is the heart of the proof and uses the potential function, 
shows that for every high stack $(\ell,S)$, we get a bunch of other ``non-low'' stacks participating with the line $\ell$. Refer to \Fig{claim4-2} for an accompanying illustration of the proof.

\begin{claim} \label{clm:blue_interval_prop} 
		Given $y \in B$, let $x := M^{-1}(y)$ and suppose $(\ell,S) \in H$ is such that $(x,y) \in M^{(\ell,S)}$ (note that this stack, $(\ell,S)$, exists by definition of $B$).
		Then, for any $z \in \cI_y \cap B$, $(\ell,S_{z}) \notin L$. 
\end{claim}


\begin{proof}
	The claim is obviously true if $z = y$, since this implies $S_{z} = S$ (since $y \in S$) and $(\ell,S) \in H$ by assumption. Therefore, we may assume $z \neq y$, and we also assume, for contradiction's sake, $(\ell,S_z)\in L$. 
	Note that $x \in \ell$ and by definition of $\cI_y$, we get $x \prec z \prec y$.

%
	Since $z\in B$, it is matched to some $w\in R$. Note $w\prec z \prec y$. Furthermore, the stack $(\ell_w,S_z) \in H$ (by definition of $B$). Thus, note that if $\ell_w = \ell$ (i.e., $w \in \ell$), then we're done and so in what follows we assume $\ell_w \neq \ell$.
	By assumption of the claim, $(\ell,S)\in H$. In particular, $x,w,z,y \in V(H)$. Now consider the new matching $N$ which deletes $(x,y)$ and $(w,z)$
	and adds $(x,z)$ and $(w,y)$. Note that the cardinality remains the same, i.e. $|N| = |M|$. 
	
	We now show that $\Lambda(N)$ is lexicographically bigger than $\Lambda(M)$. To see this, consider the stacks whose sizes have changed from $M$ to $N$. 
	There are four of them (since we swap two pairs), namely the stacks $(\ell,S), (\ell_w,S_z), (\ell,S_z)$, and $(\ell_w,S)$. For brevity's sake, let us denote their sizes in $M$ as $\lambda_1, \lambda_2,\lambda_3$, and $\lambda_4$, respectively. In $N$, their sizes are $\lambda_1 - 1, \lambda_2 - 1, \lambda_3 + 1$, and $\lambda_4 + 1$. Note that $\lambda_3 \leq \lambda - 2$ and both $\lambda_1$ and $\lambda_2$ are $\ge \lambda$. In particular, the ``new'' size of stack $(\ell,S_z)$ is still smaller than the ``new'' sizes of stacks $(\ell,S)$ and $(\ell_w, S_z)$. That is, 
	the vector $\Lambda(N)$, even without the increase in $\lambda_4$, is lexicographically larger than $\Lambda(M)$. Since increasing
    the smallest coordinate (among some coordinates) increases the lexicographic order, we get a contradiction to the lexicographic maximality of $\Lambda(M)$. 
\end{proof}

\begin{figure}[h]
    \begin{center}
    \includegraphics[clip, scale=0.6]{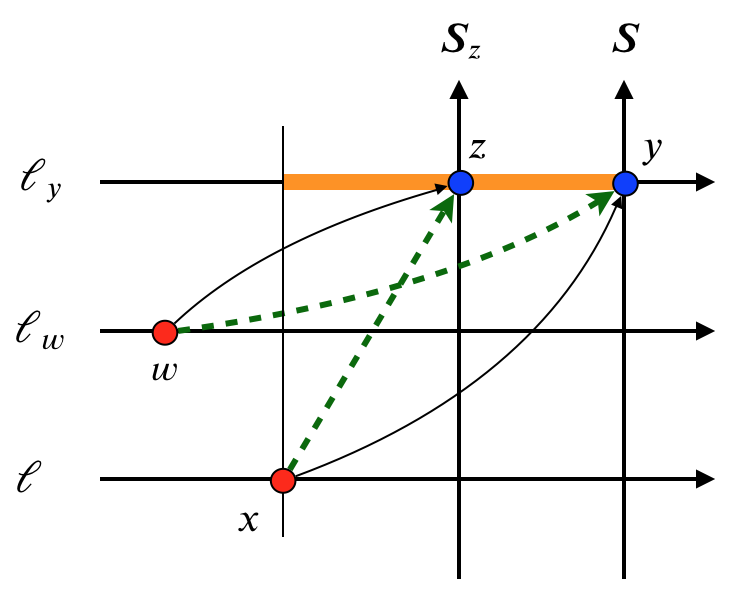}
    \end{center}
    \caption{Accompanying illustration for the proof of \Clm{blue_interval_prop}. The black connecting arrows represent the matching, $M$, while the dashed green arrows represent the new matching, $N$. The bold orange segment of $\ell_y$ is the interval $\cI_y$.}
    \label{fig:claim4-2}
\end{figure}

The rest of the proof is a (slightly technical) 
averaging argument to prove that $|B|$ is small. We introduce some more notation to carry this through.
For a blue point $y \in B$, let $\beta_y := \frac{|\cI_y \cap B|}{|\cI_y|}$ denote the fraction of blue points in $\cI_y$. For $\alpha \in (0,1)$, we say that $y \in B$ is $\alpha$-rich if $\beta_y \geq \alpha$. A point $x\in R$ is $\alpha$-rich if its blue partner $y\in B$ (i.e. $(x,y)\in M$) is $\alpha$-rich. We also call the pair $(x,y)$ an $\alpha$-rich pair. For what follows, recall that $\delta \in (0,1)$ is defined such that $|B| = \delta|D|$.

\begin{claim} \label{clm:rich_blue} At least $\delta |D|/2$ of the points in $B$ are $\delta/4$-rich.
%
\end{claim}

%

\begin{proof} 
	\def\Bpoor{B^{(\mathsf{poor})}}
	\def\Cpoor{B^{(\mathsf{poor})}_{\mathsf{min}}}
	Let $\Bpoor \subseteq B$ be the points with $\beta_y < \delta/4$. We show $|\Bpoor| \leq \delta |D|/2$ which proves the claim.	
	To see this, first observe $\Bpoor \subseteq \bigcup_{y \in \Bpoor} \left(\cI_y \cap B\right)$. 
	Now consider the {\em minimal} subset $\Cpoor \subseteq \Bpoor$ such that $\bigcup_{y \in \Cpoor} \cI_y  = \bigcup_{y \in \Bpoor} \cI_y $. That is, given a collection of intervals, we are picking the minimal subset covering the same points.
	Since these are intervals, we get that no point is contained in more than two intervals $\cI_y$ among $y\in \Cpoor$. In particular, this implies
	\begin{equation}\label{eq:trivial}
	\sum_{y\in \Cpoor} |\cI_y | \leq 2\cdot \left|\bigcup_{y \in \Cpoor} \cI_y  \right| \text{.}
	\end{equation}
	Therefore,
%
%

\begin{align*}
    \left|\Bpoor\right| &\leq \left|\bigcup_{y \in \Bpoor} \left(\cI_y \cap B\right) \right| = \left|\bigcup_{y \in \Cpoor} \left(\cI_y \cap B\right) \right|\leq \sum_{y \in \Cpoor} |\cI_y \cap B| \\ 
    &<  \frac{\delta}{4}\sum_{y \in \Cpoor} |\cI_y| ~\leq~ \frac{\delta}{2}\cdot\left|\bigcup_{y \in \Cpoor} \cI_y\right| ~\leq~ \frac{\delta}{2} \cdot |D| \text{.}
\end{align*}
The first equality follows from the definition of $\Cpoor$ (taking intersection with $B$), and the third (strict) inequality follows from the fact that none of these points are $\delta/4$-rich. The fourth inequality is \eqref{eq:trivial}. This completes the proof. 
\end{proof}

A corollary of \Clm{rich_blue} is that there are at least $\delta |D|/2$ red points which are $\delta/4$-rich. In particular, 
there must exist some line $\ell$ that contains $\geq \delta n/2$ red points in it which are $\delta/4$-rich. Let this line be $\ell$ and let
$R_\ell \subseteq \ell$ be the set of rich red points. Let $B_\ell$ be their partners in $M$.
%
%
Let $\cS^\ell = \left\{S \in \cS : \exists z \in S \cap \left(\cup_{y \in B^\ell} \cI_y \cap B\right)\right\}$ denote the set of slices containing blue points from the collection of rich intervals, $\{\cI_y : y \in B^\ell\}$. By \Clm{blue_interval_prop}, we know that all these stacks are non-low, that is, $(\ell,S) \notin L$ for all $S \in \cS^\ell$. We now {\em lower bound} the cardinality of this set.

Consider the set of blue points in our union of rich intervals from $B^\ell$, $\bigcup_{y \in B^\ell} \cI_y \cap B$. There are precisely $n$ slices in total, and for a vertex $z \in D$, $S_{z}$ is the slice indexed by the $1$-coordinate of $z$. Thus, we have $|\cS^\ell| = |\{z_1 : z \in \bigcup_{y \in B^\ell} \cI_y \cap B\}|$. That is, $|\cS^\ell|$ is exactly the number of unique $1$-coordinates among vertices in $\bigcup_{y \in B^\ell} \cI_y \cap B$. 

Since we care about the number of unique $1$-coordinates, we consider the "projections" of our sets of interest onto dimension $1$. 
For a set $X \subseteq D$, let $\bo(X) := \{x_1 : x \in X\}$ be the set of $1$-coordinates used by points in $X$. In particular, note that for $y \in B$, $\bo(\cI_y) := [x_1,y_1] \subset [n]$, where $x := M^{-1}(y)$ and observe that $|\cS^\ell| = \left|\bigcup_{y \in B^\ell} \bo(\cI_y \cap B)\right|$. Now, given that each interval from $\{\cI_y\}_{y \in B^\ell}$ is a $\frac{\delta}{4}$-fraction blue, the following claim says that at least a $\frac{\delta}{8}$-fraction of the union of intervals consists of blue points with unique $1$-coordinates. 

\begin{claim} \label{clm:blue_union} $\left|\bigcup_{y \in B^\ell} \bo(\cI_y \cap B)\right| \geq \frac{\delta}{8} \left|\bigcup_{y \in B^\ell} \bo(\cI_y)\right|$. \end{claim}

\noindent \emph{Proof.} 
As in the proof of \Clm{blue_interval_prop}, let $B_{\mathsf{min}}^\ell \subseteq B^\ell$ be a minimal cardinality subset of $B^\ell$ such that $\bigcup_{y \in B^\ell_{\mathsf{min}}} \bo(\cI_y) = \bigcup_{y \in B^\ell} \bo(\cI_y)$. For any $y \in B$, $y$ belongs to \emph{at most} two intervals from $B_{\mathsf{min}}^\ell$.

\begin{align*}
    \left|\bigcup_{y \in B^\ell} \bo(\cI_y \cap B)\right| &= \left|\bigcup_{y \in B^\ell_{\mathsf{min}}} \bo(\cI_y \cap B)\right| \geq \frac{1}{2} \sum_{y \in B^\ell_{\mathsf{min}}} \left|\bo(\cI_y \cap B)\right| \\
    &\geq \frac{\delta}{8} \sum_{y\in B^\ell_{\mathsf{min}}} \left|\bo(\cI_y)\right| \geq \frac{\delta}{8} \left|\bigcup_{y \in B^\ell_{\mathsf{min}}} \bo(\cI_y)\right| = \frac{\delta}{8} \left|\bigcup_{y \in B^\ell} \bo(\cI_y)\right| \text{.} \pushQED{\qed} \qedhere \popQED
\end{align*}

Now importantly, $|\bo(R^\ell)| = |R^\ell| \geq \frac{\delta}{2} \cdot n$ since the $1$-coordinates of elements of $R^\ell$ are distinct (since $R^\ell$ is contained on a single line). Moreover, by definition of $\cI_y$, $\bo(R^\ell) \subseteq \bigcup_{y \in B^\ell} \bo(\cI_y)$ and so $\left|\bigcup_{y \in B^\ell} \bo(\cI_y)\right| \geq |\bo(R^\ell)| \geq \frac{\delta}{2} \cdot n$. Finally, combining this with \Clm{blue_union}, we get

\begin{align*}
    |\cS^\ell| = \left|\bigcup_{y \in B^\ell} \bo(\cI_y \cap B)\right| \geq \frac{\delta}{8} \left|\bigcup_{y \in B^\ell} \bo(\cI_y)\right| \geq \frac{\delta^2}{16} \cdot n \text{.}
\end{align*}

Therefore, $\ell$ participates in at least $\frac{\delta^2}{16} \cdot n$ non-low stacks. Thus, by \Clm{trivial_bound}, $\frac{\delta^2}{16} \cdot n \leq \frac{n}{\lambda - 1}$ and so $\delta \leq \frac{4}{\sqrt{\lambda-1}}$.
Since $\lambda > 100$, 
we conclude that $\delta \leq \frac{5}{\sqrt{\lambda}}$. This concludes the proof of \Lem{stack_bound}. \pushQED{\qed} \qedhere \popQED

\section{Line Sampling: Proof of \Lem{line_sampling}} \label{sec:discrepancy}

We recall the lemma for ease of reading. Given a line $\ell \in \cL$, we have defined $M^{(\ell)} := \{(x,y) \in M : x \in \ell\}$.
Given a stack $S$, we have defined $M^{(\ell,S)} := \{(x,y)\in M^{(\ell)}: y\in S\}$.
Given a multi-set $T \subseteq [n]$, recall $M^{(\ell)}_T$ is a maximum cardinality matching of violations $(x,y)$ such that 
(a) $x$ and $y$ are both matched by $M^{(\ell)}$, and (b) $x_1$ and $y_1$ both lie in $T$. Given $\lambda \in \mathbb{Z}^+$ such that $|M^{(\ell,S)}| \leq \lambda$ for all $\ell \in \cL$ and $S \in \cS$, the line sampling lemma (\Lem{line_sampling}) states

\begin{align}
	\Expect_{T}\left[|M^{(\ell)}_T|\right] \geq \frac{k}{n} \cdot |M^{(\ell)}| - 3 \lambda \sqrt{k \ln k} \text{.}
\end{align}

We note that BRY (Theorem 3.1,~\cite{BeRaYa14}) prove a stronger theorem for the $\lambda = 1$ case
(that gets an additive error of $\Theta(\sqrt{k})$). Our proof follows a similar approach.


	

Consider an arbitrary, fixed line $\ell \in \cL$. We use the matching $M^{(\ell)}$ to induce weights $w^+(i),w^-(i)$ on $[n]$ as follows.
Initially $w^+(i),w^-(i) = 0$ for all $i \in [n]$.
For each $(x,y)\in M^{(\ell)}$ if $x\in S_i$ then we increase $w^+(i)$ by $1$, and if $y\in S_j$ then we increase $w^-(j)$ by $1$. 

\begin{claim}\label{clm:easy}
We make a few observations.
	\begin{enumerate}[noitemsep]
\item	For any $i\in [n]$, $w^+(i) \le 1$.
\item 	For any $i\in [n]$, $w^-(i) \leq \lambda$.
\item   For any $t\in [n]$, $\sum_{s \leq t} (w^-(s) - w^+(s)) \leq 0$.
	\end{enumerate}
\end{claim}
\begin{proof}
The first observation follows since the lower endpoints of $M^{(\ell)}$ all lie on $\ell$, and thus have distinct $1$-coordinates.
The second observation follows from the assumption that $|M^{(\ell,S)}| \leq \lambda$ for all $(\ell,S) \in \cL \times \cS$.
The third observation follows by noting that whenever $w^-(j)$ is increased for some $j$, we also increase $w^+(i)$ for some $i<j$.
\end{proof}
\noindent
Define $V^+ := \{i: w^+(i) > 0\}$ and $V^- := \{j: w^-(j) > 0\}$.
Given a multiset $T\subseteq [n]$, denote $V^+_{T} := V^+ \cap T$ and $V^{-}_{T} := V^- \cap T$.
Also, define the bipartite graph $G_T := (V^+_{T}, V^-_{T}, E_T)$ where $(i,j) \in E_T$ iff $i\leq j$.
A $w$-matching $A$ in $G_T$ is a subset of edges of $E_T$ such that every vertex $i \in V^+_{T}$ has at most $w^+(i)$ edges of $A$ incident on it,
and every vertex $j\in V^-_{T}$ has at most $w^-(j)$ edges of $A$ incident on it.
Let $\nu(G_T)$ denote the size of the largest $w$-matching in $G_T$.
\begin{lemma}\label{lem:apple}
For any multiset $T\subseteq [n]$ and any $w$-matching $A\subseteq E_T$ in $G_T$, we have $|M^{(\ell)}_{T}| \geq |A|$. 
In particular, $\Expect_{T}\left[|M^{(\ell)}_T|\right] \geq \Expect_{T}\left[\nu(G_T)\right]$.
\end{lemma}
\begin{proof}
Consider any $w$-matching $A\subseteq E_T$. For any vertex $i\in V^+_{T}$, there are at most $w^+(i)$ edges in $A$ incident on it.
Each increase of $w^+(i)$ is due to an edge $(x,y)\in M^{(\ell)}$ where $x_1 = i$. Thus, we can charge each of these edges of $A$ (arbitrarily, but uniquely) to 
$w^+(i)$ different $x\in \ell$. 
Similarly, for any vertex $j\in V^-_{T}$, there are at most $w^-(j)$ edges in $A$ incident on it.
Each increase of $w^-(j)$ is due to an edge $(x,y)\in M^{(\ell)}$ with $y_1 = j$. Thus, we can charge each of these edges of $A$ (arbitrarily, but uniquely) to $w^-(j)$ different $y\in S_j$, the $j$th slice. Furthermore, any $z \in \ell$ with $z_1 \leq j$ satisfies $z \prec y$.
To summarize, each $(i,j) \in A$ can be uniquely charged to an $x\in \ell$ with $x_1 = i$ and $y\in S_j$ such that (a) $(x,y)$ forms a violation, (b) $x,y$ were matched in $M^{(\ell)}$, and (c) $x_1,y_1\in T$. Therefore, $|M^{(\ell)}_T| \ge |A|$ since the LHS is the maximum cardinality matching.
\end{proof}

\begin{lemma}~\label{lem:match}
For any $T\subseteq [n]$, we have 
\[
\nu(G_T) = \sum_{j\in T} w^-(j) - \max_{t \in T} \sum_{s \in T : s\leq t} \left(w^-(s) - w^+(s) \right)
 \text{.} \]
\end{lemma}

\begin{proof} By Hall's theorem, the maximum $w$-matching in $G_T$ is given by the total weight on the $V^-_{T}$ side, that is, $\sum_{j\in T} w^-(j)$, 
minus the total {\em deficit} $\delta(T) := \max_{S \subseteq V^-_T} \left(\sum_{s \in S} w^-(s) - \sum_{s \in \Gamma_T(S)}w^+(s)\right)$ where for $S \subseteq V^-_T$, $\Gamma_T(S) \subseteq V^+_T$ is the neighborhood of $S$ in $G_T$. Consider such a maximizer $S$, and let $t$ be the largest index present in $S$.
Then note that $\sum_{s \in \Gamma_T(S)}w^+(s)$ is precisely $\sum_{s \in T : s\leq t} w^+(s)$. Furthermore note that adding any $s\leq t$ from $V^-_T$ won't increase
$|\Gamma_T(S)|$. Thus, given that the largest index present in $S$ is $t$, we get that $\delta(T)$ is precisely the summation in the second term of the RHS. $\delta(T)$ is maximized by choosing the $t$ which maximizes the summation. \end{proof}


Next, we bound the expectation of the RHS in \Lem{match}. Recall that $T := \{s_1,\ldots,s_k\}$ is a multiset where each $s_i$ is u.a.r. picked from $[n]$. For the first term, we have 

\begin{equation}\label{eq:easy}
\Expect_{T}\left[\sum_{j\in T} w^-(j)\right] = \sum_{i=1}^k \sum_{j=1}^n \Pr[s_i = j]\cdot w^-(j) = \frac{k}{n}\cdot \sum_{j=1}^n w^-(j) = \frac{k}{n}\cdot |M^{(\ell)}| \text{.}
\end{equation}
The second-last equality follows since $s_i$ is u.a.r. in $[n]$ and the last equality follows since $\sum_j w^-(j)$ increases by exactly one for each edge in $M^{(\ell)}$. Next we upper bound the expectation of the second term.
For a fixed $t$, define 
	\[
		Z_t := \sum_{s \in T : s\leq t} (w^-(s) - w^+(s)) = \sum_{i = 1}^k  X_{i,t}~~\textrm{where}~ X_{i,t} = \begin{cases} w^-(s_i) - w^+(s_i) & \textrm{if $s_i \leq t$} \\ 0 & \textrm{otherwise}\end{cases} \text{.}
	\]
 Note that the $X_{i,t}$'s are i.i.d. random variables with $X_{i,t} \in [-1,\lambda]$ with probability $1$. 
Thus, applying Hoeffding's inequality we get

\begin{align} \label{eq:hoeffding}
	\Pr\left[Z_t > \Exp[Z_t] + a\right]	 \leq 2 \exp\left(\frac{-a^2}{2k\lambda^2}\right) \text{.}
\end{align}
Now we use \Clm{easy}, part (3) to deduce that 
\[
\Exp[Z_t] = \sum_{i=1}^k \Exp[X_{i,t}] = \sum_{i=1}^k \sum_{s \le t} (w^-(s) - w^+(s))\cdot \Pr[s_i = s] \leq 0 
\] 
since $\Pr[s_i = s] = 1/n$. Therefore, the RHS of \Eqn{hoeffding} is an upper-bound on $\Pr[Z_t \geq a]$. 
In particular, invoking $a := 2 \lambda \sqrt{k\ln k}$ and applying a union bound, we get
\begin{equation*}
\Pr\left[\max_{t\in T} Z_t > 2\lambda\sqrt{k\ln k}\right] = \Pr\left[\exists t \in T:~~ Z_t > 2\lambda\sqrt{k\ln k}\right] \leq k \cdot e^{-2 \ln k} = 1/k
\end{equation*}

\noindent and since $\max_{t \in T} Z_t$ is trivially upper-bounded by $\lambda k$, this implies that 

\begin{align}
\label{eq:easy2}
\Expect_{T} \left[ \max_{t \in T} \sum_{s \in T : s\leq t} \left(w^-(s) - w^+(s) \right)\right] \leq \lambda k \cdot \Pr\left[\max_{t \in T} Z_t > a\right] + a \leq \lambda + a \leq 3 \lambda \sqrt{k\ln k} \text{.}
\end{align}

\noindent \Lem{line_sampling} follows from \Lem{apple}, \Lem{match}, \Eqn{easy}, and \Eqn{easy2}.

\section{The Continuous Domain} \label{sec:discrete}

We start with measure theory preliminaries.
We refer the reader to Nelson~\cite{Nelson} and Stein-Shakarchi~\cite{StSh-book} for more background.
Given two reals $a < b$, we use $(a,b)$ to denote the open interval, and $[a,b]$ to denote the closed interval.
Given $d$ closed intervals $[a_i,b_i]$ for $1\leq i\leq d$, we call their Cartesian product $\prod_{i \in [d]} [a_i,b_i]$ a {\em box}. Two intervals/boxes are \emph{almost disjoint}
if their interiors are disjoint (they can intersect only at their boundary).
An \emph{almost partition} of a set $S$ is a collection $\cP$ of sets that are pairwise almost disjoint and $\bigcup_{P \in \cP} P = S$.
A set $U$ is \emph{open} if for each point $x\in U$, there exists an $\eps > 0$ such that the sphere centered at $x$ of radius $\eps$ is contained in $U$. 

We let $\mu = \prod_{i \in [d]} \mu_i$ be an arbitrary \emph{product measure} over $\RN^d$. That is, each $\mu_i$ is described
by a non-negative Lebesgue integrable function over $\RR$, whose total integral is $1$ (this is the pdf). 
%
	Abusing notation, we use $\mu_i([a_i,b_i]) = \Pr_{x \sim \mu_i}[a_i \leq x \leq b_i]$ to denote the integral of $\mu_i$ over this interval. Indeed, this is the probability measure of the interval. 
	The volume of a box $B = \prod_{i \in [d]} [a_i,b_i]$ is denoted $\mu(B) = \prod_{i \in [d]} \mu_i([a_i,b_i]) = \Pr_{x \sim \mu}[x \in B]$. 

%
%

We use the definition of measurability of Chapter 1.1.3 of~\cite{StSh-book}. Technically, this is given with respect to
the standard notion of volume in $\RR^d$. Chapter 6, Lemma 1.4 and Chapter 6.3.1 show that the definition
is valid for the notion of volume with respect to $\mu$, as we've defined above. The \emph{exterior measure} $\mu_*$ of any set $E$ is the infimum of the sum of
volumes of a collection of closed boxes that contain $E$. 


\begin{definition} \label{def:leb-measure} Given a product measure $\mu = \prod_i \mu_i$ over $\RR^d$, we say $E \subseteq \RR^d$ is Lebesgue-measurable with respect to $\mu$ if for any $\eps > 0$, there exists an open set $U \supseteq E$ such that $\mu_*(U \setminus E) < \eps$. If this holds, then the $\mu$-measure of $E$ is defined as $\mu(E) := \mu_*(E)$. \end{definition} 

Given a function $f\colon \RN^d \to \{0,1\}$, we will often slightly abuse notation by letting $f$ denote the \emph{set} it indicates, i.e. the set in $\RR^d$ where $f$ evaluates to $1$. We say that $f$ is a \emph{measurable function w.r.t. $\mu$} if this set is measurable w.r.t. $\mu$. Similarly, we use $\overline{f}$ to denote the \emph{set} where $f$ evaluates to $0$. 

We are now ready to define the notion of distance between two functions. In \Sec{mono-measure}, we prove that all monotone Boolean functions are measurable (\Thm{mono-measure}) with respect to $\mu$. Also, measurability is closed under basic set operations and thus the following notion of distance to monotonicity is well-defined.

\begin{definition} [Distance to Monotonicity]
	Fix a product measure $\cD$ on $\RR^d$. 
	We define the distance between two measurable functions $f,g\colon \RN^d \to \{0,1\}$ with respect to $\mu$, as 
	\begin{align}
	\dist_{\cD}(f,g) := \mu\left(\left\{z \in \RN^d : f(z) \neq g(z)\right\}\right) = \mu\left(f \Delta g\right) \text{.}
	\end{align}
The distance to monotonicity of $f$ w.r.t. $\cD$ is defined as	
	\begin{align} 
	\eps_{f,\cD} := \inf_{g \in \cM} \dist_{\cD}(f,g) = \inf_{g \in \cM} \mu\left(f \Delta g\right)
	\end{align}
	
	\noindent where $\cM$ denotes the set of monotone Boolean functions over $\RN^d$.
\end{definition}

We are now equipped to state the formal version of \Thm{main-informal-cont}, for testing Boolean functions over $\RN^d$. 

\begin{theorem}\label{thm:main-cont}
	Let $\cD = \prod_{i=1}^d \cD_i$ be a product measure 
	for which we have the ability to take independent samples from each $\cD_i$.
	There is a randomized algorithm which, given a parameter $\eps > 0$ and a measurable function $f\colon \RR^d\to \{0,1\}$ that can be queried at any $x \in \RR^d$, makes $\tilde{O}(d^{5/6}\eps^{-4/3})$ non-adaptive queries to $f$, and (a) always accepts if $f$ is monotone, and (b) rejects with probability $> 2/3$ if $\eps_{f,\cD} > \eps$.
\end{theorem}

We give a formal proof of \Thm{main-cont} in \Sec{tester}. The proof requires some tools to discretize measurable sets, which we provide in the next two sections. 

\subsection{Approximating measurable sets by grids} \label{sec:lemma}

	We first start with a lemma about probability measures over $\RR$. 
	\begin{lemma}\label{lem:cont-equal-int}
	Given any probability measure $\cD$ over $\RR$, and any $N\in \NN$, there exists an 
    almost partition of $\RR$ into
	$N$ intervals $\bI_N = \{\cI_1, \ldots, \cI_N\}$ of equal $\mu$-measure. That is, for each $j\in [N]$, $\Pr_{x \sim \cD}[x\in \cI_j] = \frac{1}{N}$. Furthermore, for any $k \in \NN$, 
    $\bI_{kN}$ is a refinement of $\bI_N$.
	\end{lemma}
	\begin{proof}
		$\mu$ is a probability measure, and thus is described by a non-negative Lebesgue-integrable function (it's pdf). Chapter 2, Prop 1.12 (ii) of \cite{StSh-book} states that the Lebesgue integral is continuous and thus it's CDF, $F(t) := \mu(\{x\in \RR: x \leq t\})$, is continuous. Moreover $F$ is non-decreasing with range $[0,1]$. Therefore,  
		for every $\theta \in (0,1)$ there is at least one $t$ with $F(t) = \theta$. Thus, let's define $F^{-1}(\theta)$ to be the supremum over all $t$ satisfying $F(t) = \theta$. Let $F^{-1}(0) = -\infty$ and $F^{-1}(1) = +\infty$. 
		The lemma is proved by the intervals $\cI_j = [F^{-1}((j-1)/N), F^{-1}(j/N)]$ for $j\in\{1,\ldots,N\}$.
        The refinement is evident by the fact that any interval in $\bI_N$ can be
        expressed as an almost partition of intervals from $\bI_{kN}$ (for $k \in \NN$).
		%
	\end{proof}

Thus, given a product distribution $\cD = \prod_{i=1}^d \cD_i$ and any $N\in \NN$, we can apply the above lemma to each of the $d$ coordinates to obtain the set of $Nd$ intervals $\left\{\cI^{(i)}_j : i\in [d] : j\in [N]\right\}$ for which $\mu_i\left(\cI^{(i)}_j\right) = 1/N$ for every $i \in [d]$, $j \in [N]$. We define

\[\bG_N := \left\{\prod_{i=1}^d \cI^{(i)}_{z_i}: z \in [N]^d\right\}\]

\noindent and observe that (a) $\bG_N$ is an almost partition of $\RR^d$ and 
(b) $\bG_{kN}$ is a refinement of $\bG_{N}$ for any $k \in \NN$. (Since $d$ is fixed, we will not carry the dependence on $d$.)
We informally refer to $\bG_{N}$ as a \emph{grid}.
Since $\bG_{N}$ is an almost partition, we can define the function
$\bxx_N: \RR^d \to [N]^d$ as follows. For $x \in \RR^d$, we define
$\bxx_N(x)$ to be the lexicographically least $z \in [N]^d$ such that the box $\prod_{i=1}^d \cI^{(i)}_{z_i}$, of $\bG_N$, contains $x$.
(Note that for all but a measure zero set, points in $\RR^d$ are contained
in a unique box of $\bG_{N}$.)

In the following lemma, we show that any measurable set can be approximated by a sufficiently fine grid. In some sense, this \emph{is} the definition of measurability.

\def\an{a^{(r)}}
\def\bn{b^{(r)}}
\def\dn{\delta^{(r)}}

\begin{lemma} \label{lem:approx-grid} For any measurable set $E$
and any $\alpha > 0$, there exists $N = N(E,\alpha) \in \NN$ such that
there is a collection $\bB \subseteq \bG_{N}$ satisfying
$\mu(E \ \Delta \ \bigcup_{B \in \bB} B) \leq \alpha$.
\end{lemma}

\begin{proof} Chapter 1, Theorem 3.4 (iv) of~\cite{StSh-book} states
that for any measurable set $E$ and any $\epsilon > 0$, there exists a finite union 
$\bigcup_{r=1}^m B_r$ of closed boxes such that $\mu(E \Delta \bigcup_{r=1}^m B_r) \leq \epsilon$. We invoke this theorem with $\epsilon = \alpha/2$ to get the collection of boxes $B_1,\ldots,B_m$. Note that these boxes may intersect, and might not form a grid.
We build a grid by setting $N = \ceil{2md/\alpha}$ and  considering $\bG_{N}$. The desired collection $\bB \subseteq \bG_{N}$
is the set of boxes in $\bG_{N}$ contained in $\bigcup_{r=1}^m B_r$. Observe that

\begin{align} \label{eq:grid-approx-diff-bound}
	\mu\left(E \Delta \bigcup_{B \in \bB} B\right) \leq \mu\left(E \Delta \bigcup_{r=1}^m B_r\right) + \mu\left(\bigcup_{r=1}^m B_r \setminus \bigcup_{B \in \bB} B\right) \leq \alpha/2 + \sum_{r=1}^m \mu\left(B_r \setminus \bigcup_{B \in \bB} B\right)
\end{align}

\noindent by subadditivity of measure. We complete the proof by bounding $\mu(B_r \setminus \bigcup_{B \in \bB} B)$ for an arbitrary $r \in [m]$.

%

Let $B_r := \prod_{i=1}^d [a_i,b_i]$ denote an arbitrary box from $\{B_1,\ldots,B_m\}$ and let $\delta_i := \mu_i([a_i,b_i])$. 
%
Observe that the interval $[a_i,b_i]$ contains exactly $\lfloor \delta_i N \rfloor$ contiguous intervals from the almost partition $\{\cI^{(i)}_{j} : j \in [N]\}$ of $\RR$. Let $\bI_i$ denote the set of such intervals.
Thus, $\mu_i([a_i,b_i] \setminus \bigcup_{I \in \bI_i} I) \leq \delta_i - (1/N)\left(\lfloor \delta_i N \rfloor\right) \leq \delta_i - (1/N)\left(\delta_i N - 1\right) = 1/N$. Thus, the total measure of $B_r$ we discard is $\mu(B_r \setminus \bigcup_{B \in \bB} B) \leq \prod_{i} \delta_i - \prod_{i} (\delta_i - 1/N)$. This quantity is maximized when the $\delta_i$'s are maximized; since $\delta_i \leq 1$ (each $\mu_i$ is a probability measure), we get that $\mu(B_r \setminus \bigcup_{B \in \bB} B) \leq 1 - (1-1/N)^d \leq \frac{d}{N}$. 

Finally, plugging this into \Eqn{grid-approx-diff-bound}, we get $\mu(E \Delta \bigcup_{B \in \bB} B) \leq \alpha/2 + m \cdot \frac{d}{N} \leq \alpha$, since $N \geq 2md/\alpha$. \end{proof}



We are now ready to prove our main tool, the discretization lemma.

\begin{lemma} [Discretization Lemma] \label{lem:good-approx} Given a measurable function $f\colon \RN^d \to \{0,1\}$ and $\delta > 0$, there exists $N := N(f,\delta) \in \NN$, 
and a function $f^{\mathsf{disc}}:[N]^d \to \{0,1\}$, such that 
$\Pr_{x \sim \cD}[f(x) \neq f^{\mathsf{disc}}(\bxx_N(x))] \leq \delta$. \end{lemma}

\begin{proof} By assumption, $f$ and $\overline{f}$ are measurable sets.
By \Lem{approx-grid}, there exists some $N_1$ and a collection of boxes
$\bZ_1 \subseteq \bG_{N_1}$ such that $\mu(f \Delta \bigcup_{B \in \bZ_1} B) \leq \delta/6$.
(An analogous statement holds for $\overline{f}$, with some $N_0$ and a collection $\bZ_0$.)
Since \Lem{approx-grid} also holds for any refinement of the relevant grid,
let us set $N = N_0N_1$. Abusing notation, we have two collections $\bZ_0, \bZ_1 \subseteq \bG_{N}$
such that $\mu(f \Delta \bigcup_{B \in \bZ_1} B) \leq \delta/6$
and $\mu(\overline{f} \Delta \bigcup_{B \in \bZ_0} B) \leq \delta/6$.

For convenience, let us treat the boxes in $\bZ_0 \cup \bZ_1$ as open, so
that all boxes in the collection are disjoint.
Define $h\colon \RN^d \to \{0,1\}$ as follows: 





\[
h(x) = 
\begin{cases}
1 & \text{ if } x \in \bigcup_{B \in \bZ_1 \setminus \bZ_0} B \\
0 & \text{ if } x \in \bigcup_{B \in \bZ_0 \setminus \bZ_1} B \\
0 & \text{ if } x \in \bigcup_{B \notin \bZ_0 \Delta \bZ_1} B\\
\end{cases} \text{.}
\]

Since $f$ and $\overline{f}$ partition $\RN^d$, $\mu(\bigcup_{B \in \bZ_0 \cap \bZ_1} B)$ and $\mu(\bigcup_{B \notin \bZ_0 \cup \bZ_1} B)$ are both at most $\mu(f \Delta \bigcup_{B \in \bZ_1} B) + \mu(\overline{f} \Delta \bigcup_{B \in \bZ_0} B) \leq \delta/3$. Combining these bounds, we have $\mu(\bigcup_{B \notin \bZ_0 \Delta \bZ_1} B) \leq 2\delta/3$. Thus


\noindent 
	
\begin{align*}
	\dist_{\cD}(f,h) = \Pr_{x \sim \cD}[f(x) \neq h(x)] 
    &\leq \mu\left(\bigcup_{B \in \bZ_1 \setminus \bZ_0} B \cap \overline{f}\right) + \mu\left(\bigcup_{B \in \bZ_0 \setminus \bZ_1} B\cap f\right) + \mu\left(\bigcup_{B \notin \bZ_0 \Delta \bZ_1} B\right) \\ 
    &\leq \delta/6 + \delta/6 + 2\delta/3 = \delta \text{. } \hspace{4mm}
\end{align*}


By construction, $h$ is constant in (the interior of) every grid box.
Any $z \in [N]^d$ indexes a (unique) box in $\bG_N$ (recall the map $\bxx_N\colon \RR^d \to [N]^d$).
Formally, we can define a function $f^{\mathsf{disc}}\colon [N]^d \to \{0,1\}$ so that 
$\forall x \in \RR^n, f^{\mathsf{disc}}(\bxx_N(x)) = h(x)$.
Thus, $\Pr_{x \sim \cD}[f(x) \neq f^{\mathsf{disc}}(\bxx_N(x))]  = \dist_{\cD}(f,h) \leq \delta$.
\end{proof}

\subsection{Proof of \Thm{cont_domain_reduction}} \label{sec:proof_cont_domain_reduction}

	\begin{proof}
		Recall that $\bT = T_1 \times \cdots \times T_d$ is a randomly chosen hypergrid, where for each $i \in [d]$, $T_i \subset \RN$ is formed by taking $k$ i.i.d. samples from $\cD_i$.
		We need to show that
			
			\[\Expect_{\bT}\left[\eps_{f_{\bT}}\right] \geq \eps_f - \frac{C' \cdot d}{k^{1/7}}\]
			
			\noindent for some universal constant $C' > 0$. 
	
%
	
	Set $\delta \leq k^{-d}\cdot \frac{C\cdot d}{k^{1/7}}$, where $C$ is the universal constant in \Thm{dir_domain_reduction}.
	Applying \Lem{good-approx} to $f$ with this $\delta$, we know there exists $N > 0$ and $f^{\mathsf{disc}}\colon [N]^d \to \{0,1\}$, such that $\Pr_{x \sim \cD}[f(x) \neq f^{\mathsf{disc}}(\map_N(x))] \leq \delta$. 
	
	Given a random $\bT$ sampled as described above, define 
	$\widehat{\bT} := \{\map_N(x) \in [N]^d : x \in \bT\}$. 
	Observe that (a) $\widehat{\bT}$ is a $[k]^d$ sub-hypergrid in $[N]^d$ which (b) can be equivalently defined as $\widehat{\bT} = \widehat{T}_1 \times \cdots \times \widehat{T}_d$ where each $\widehat{T}_i$ is formed by taking $k$ i.i.d. uniform samples from $[N]$. This is by construction of the partition $\{\bxx_z : z \in [N]^d\}$ and by definition of $\map_N(x)$. 
	\Thm{dir_domain_reduction} and the observations above imply 
	\begin{equation}\label{eq:0021}
	\Expect_{\widehat{\bT}}\left[\eps_{f^{\mathsf{disc}}_{\widehat{\bT}}}\right] \geq \eps_{\fd} - \frac{C\cdot d}{k^{1/7}}
	\end{equation}
	where $C$ is some universal constant.
	Next, we relate $\eps_{\fd}$ and $\eps_f$.
	Observe that there is a bijection between $\bT$ and $\widehat{\bT}$ (namely, $\map_N$ restricted to $\bT$). We say $f_{\bT} = f^{\mathsf{disc}}_{\widehat{\bT}}$ if for all $x\in \bT$, $f(x) = \fd(\map_N(x))$.	
	
	By a union bound over the $k^d$ samples, 
	\[\Pr_{\bT}\left[f_{\bT} \neq f^{\mathsf{disc}}_{\widehat{\bT}}\right] = \Pr_{\bT}\left[\exists x \in \bT :  f(x) \neq \fd(\map_N(x)) \right] \leq \delta \cdot k^d \leq \frac{C\cdot d}{k^{1/7}} =:\delta'
	\]
	since each $x \in \bT$ has the same distribution as $x\sim \cD$, and $\Pr_{x \sim \cD}[f(x) \neq f^{\mathsf{disc}}(\map_N(x))] \leq \delta$. 
	Thus, we get $\Expect_{\bT}\left[\eps_{f_{\bT}}\right] \geq (1-\delta') \Expect_{\widehat{\bT}}\left[\eps_{\fd_{\widehat{\bT}}}\right] - \delta'$, since in the case
	$f_{\bT} \neq \fd_{\widehat{\bT}}$, the difference in their distance to monotonicity is at most $1$. Substituting in \eqref{eq:0021}, we get
	\begin{equation}\label{eq:0022}
		\Expect_{\bT}\left[\eps_{f_{\bT}}\right] \geq (1-\delta')\cdot\left(\eps_{\fd} - \frac{C\cdot d}{k^{1/7}}\right) - \delta' \geq \eps_{\fd} - \frac{3C\cdot d}{k^{1/7}} 
	\end{equation}
	by definition of $\delta'$.
	
	Now, let $g\colon [N]^d \to \{0,1\}$ be any monotone function satisfying $d(f^{\mathsf{disc}},g) = \eps_{f^{\mathsf{disc}}}$.
	Define the monotone function $\hat{f}(x) = g(\map_N(x))$ for all $x \in \RN^d$. 
	Note that $\eps_f \leq \dist(f, \hat{f}) \leq \Pr_{x \sim \cD}[f(x)\neq \fd(\map_N(x))] + \dist(\fd, g) \leq \delta + \eps_{\fd}$.
	This, in turn, implies $\eps_{\fd} \geq \eps_f - \delta \geq \eps_f - \frac{C\cdot d}{k^{1/7}}$.
	Substituting in \eqref{eq:0022}, we get
	
	\[
	\Expect_{\bT}\left[\eps_{f_{\bT}}\right] \geq \eps_f - \frac{4C \cdot d}{k^{1/7}}
	\]
	which proves the theorem.
	\end{proof}

\subsection{Measurability of Monotone Functions} \label{sec:mono-measure}

\begin{theorem} \label{thm:mono-measure} 
	Monotone functions $f\colon \RR^d \rightarrow \{0,1\}$ are measurable w.r.t. product measures $\mu = \prod_{i=1}^d \mu_i$.
\end{theorem}

\begin{proof} The proof is by induction over the number of dimensions, $d$. For $d=1$,
	the set $f$ is either $[z,\infty)$ or $(z,\infty)$ for some $z \in \RR$, since $f$ is a monotone function. Any open or closed set
	is measurable.
	
	Now for the induction. Choose any $\eps > 0$. We will construct
	an open set $\cO$ such that $\mu_*(\cO \setminus f) \leq 8\eps$. 
	Consider the first dimension, and the corresponding measure $\mu_1$.
	We use $\mu_{-1}$ for the $(d-1)$-dimensional product measure in the remaining dimensions.
	(We use $\mu_{-1,*}$ for the $(d-1)$-dimensional exterior measure.)
	As shown in \Lem{cont-equal-int}, there is an almost partition
	of $\RR$ into $N = \lceil 1/\eps^2\rceil$ closed intervals such that each interval has $\mu_1$-measure at most $\eps^2$. Let these intervals
	be $I_1, I_2, I_3, \ldots, I_N$. We will consider the set of intervals $\bI = \{I_1 \cup I_2, I_2 \cup I_3, \ldots, I_{N-1} \cup I_{N}\}$ (let
	us treat these as open intervals). Observe that $\cup_{I \in \bI} I = \RR$,
	and $\mu_1(I) \leq 2\eps^2$ for all $I \in \bI$. 
	
	For any $x \in \RR$, let $S_x$ be the subset of $f$ with first
	coordinate $x$. We will treat $S_x$ as a subset of $\RR^{d-1}$ and use $\{x\} \times S_x$ to denote the corresponding subset of $\RR^d$.
	By monotonicity, $\forall x < y$, $S_{x} \subseteq S_{y}$.
	By induction, each set $S_{x}$ is measurable in $\RR^{d-1}$ 
	and thus there exists an open set $\cO_x \subseteq \RR^{d-1}$
	such that $\mu_{-1,*}(\cO_x \setminus S_{x}) \leq \eps$. Define
	the function $h\colon \RR \to [0,1]$ such that $h(x)$ is the measure of $S_{x}$ (in $\RR^{d-1}$).
	Crucially, $h$ is monotone because $f$ is monotone.
	
	Call an interval $(x,y)$ \emph{jumpy} if $h(y) > h(x) + \eps$
	and let $\bJ \subseteq \bI$ be the set of jumpy intervals in $\bI$. 
	For a non-jumpy interval $I = (x,y) \in \bI \setminus \bJ$, define $\cO_I := I \times \cO_y$. Note that $\cO_I$ is open and by monotonicity, $\cO_I \supseteq \bigcup_{z \in I} (\{z\} \times S_z) = \{z \in f : z_1 \in I\}$.
	
	The open set $\cO := (\bigcup_{J \in \bJ} J \times \RR^{d-1}) \cup (\bigcup_{I \in \bI \setminus \bJ}
	\cO_I)$ contains (the set) $f$. It remains to bound
	
	\begin{align}
	\mu_*(\cO \setminus f) &\leq \mu_*\left(\bigcup_{J \in \bJ} J \times \RR^{d-1}\right) + \mu_*\left(\bigcup_{I \in \bI \setminus \bJ} \cO_I \setminus f\right) \nonumber \\
	&\leq \sum_{J \in \bJ} \mu_1(J) + \sum_{I \in \bI \setminus \bJ} \mu_*(\cO_I \setminus f) \leq 2\eps^2|\bJ| + \sum_{I \in \bI \setminus \bJ} \mu_*(\cO_I \setminus f) \text{.} \label{eq:diff}
	\end{align}

	To handle the first term, note that there are at least
	$|\bJ|/2$ disjoint intervals in $\bJ$ and each such interval represents
	a jump of at least $\eps$ in the value of $h$. Thus, $|\bJ|/2 \leq 1/\eps$ and so $|\bJ| \leq 2/\eps$.
	
	Now, consider $I = (x,y) \in \bI \setminus \bJ$. We have $\cO_I = I \times \cO_y$.
	By monotonicity $\cO_I \setminus f \subseteq \cO_I \setminus (I \times S_{x}) = (I \times \cO_y) \setminus (I \times S_{x})
	= I \times (\cO_y \setminus S_{x})$. Since $S_{y} \supseteq S_{x}$, $\cO_y \setminus S_x
	= (\cO_y \setminus S_y) \cup (S_y \setminus S_x)$. By sub-additivity
	of exterior measure, $\mu_{-1,*}(\cO_y \setminus S_x) \leq \mu_{-1,*}(\cO_y \setminus S_y)
	+ \mu_{-1,*} (S_y \setminus S_x)$. The former term is at most $\eps$, by the choice of $\cO_y$.
	Because $I$ is not jumpy, the latter term is $h(y) - h(x) \leq \eps$. Thus, 
	
	\begin{align*}
	\sum_{I \in \bI \setminus \bJ} \mu_*(\cO_I \setminus f) \leq \sum_{I \in \bI \setminus \bJ} \mu_1(I) \cdot (\mu_{-1,*}(\cO_y \setminus S_y) + \mu_{-1,*} (S_y \setminus S_x)) \leq 2\eps \sum_{I \in \bI \setminus \bJ} \mu_1(I) \leq 4\eps \text{.}
	\end{align*}
	
	All in all, we can upper bound the expression in \Eqn{diff} by $2\eps^2(2/\eps) + 4\eps = 8\eps$. \end{proof}
\section{The Monotonicity Tester} \label{sec:tester}

In this section we prove our main monotonicity testing results, \Thm{main-informal-grid} and \Thm{main-informal-cont} (recall the formal statement, \Thm{main-cont}). 
We use the following theorem of \cite{BlackCS18} on monotonicity testing for Boolean functions over $[n]^d$. 


\begin{theorem}[Theorem 1.1 of \cite{BlackCS18}]\label{thm:soda_result}
	There is a randomized algorithm which, given a parameter $\eps \in (0,1)$ and a function $f\colon [n]^d \to \{0,1\}$, makes $O(d^{5/6}\cdot \log^{3/2}d\cdot (\log n + \log d)^{4/3} \cdot \eps^{-4/3})$ non-adaptive queries to $f$ and (a) always accepts if $f$ is monotone, and (b) rejects with probability $>2/3$ if $\eps_f > \eps$.
\end{theorem} 

We refer to the tester of \Thm{soda_result} as the \pathtester. Using this result along with our domain reduction theorems
\Thm{dir_domain_reduction} and \Thm{cont_domain_reduction}, we design testers for Boolean-valued functions over $[n]^d$ and $\RR^d$ (refer to \Alg{tester}). We restrict our attention to the $\RN^d$ case and prove \Thm{main-informal-cont} (that is, \Thm{main-cont}); the proof of \Thm{main-informal-grid} is analogous (and the corresponding tester is analogous to \Alg{tester})
. In what follows we let $C$ denote the universal constant from \Thm{cont_domain_reduction} and we define $L := \lceil{\log(2/\eps)}\rceil$.

\begin{remark} Our tester (\Alg{tester}) uses Levin's work investment strategy (see \cite{Go-book}, Section 8.2.4) to optimize the dependence on $\eps$. We remark that if one only cares about achieving a dependence of $\poly(1/\eps)$, then the following simpler tester suffices: invoke \Step{domain_reduce} and \Step{soda_tester} (with $\eps_\ell$ replaced by $\eps/4$) of \Alg{tester} $16/\eps$ times. By Markov's inequality and the fact that $\Expect_{\bT}[\eps_{f_{\bT}}] \geq \eps/2$, with high probability at least one of the calls to \Step{domain_reduce} will yield a reduced hypergrid $\bT$ satisfying $\eps_{f_{\bT}} \geq \eps/4$. \Step{soda_tester} will then reject the restriction $f_{\bT}$, and thus reject $f$, with high probability. This leads to an $\eps^{-7/3}$ dependence on $\eps$, as opposed to the $\eps^{-4/3}$ achieved by \Alg{tester}. \end{remark}

\begin{algorithm}
	\caption{Monotonicity Tester for $f\colon \RN^d \to \{0,1\}$. Inputs: $f$ and $\eps \in (0,1)$.} \label{alg:tester}
	\begin{algorithmic}[1]
		\State \textbf{for all} $\ell \in [L+1]$:
			\State \indent \textbf{set} $Q_\ell := \lceil\frac{32\ell^2}{2^\ell \eps}\rceil$ and $\eps_\ell := 1/2^\ell$.
			\State \indent \textbf{repeat} $Q_\ell$ times:
				\State \label{step:domain_reduce} \indent \indent Sample $\bT = T_1 \times \cdots \times T_d$ as in \Thm{cont_domain_reduction} with $k = (2C \cdot \frac{d}{\eps})^7$. 
				\State \label{step:soda_tester} \indent \indent \textbf{if} \pathtester($f_{\bT},\eps_\ell,k$) returns REJECT, \textbf{then} \Return REJECT.
		\State \Return ACCEPT.
	\end{algorithmic}
\end{algorithm}

\paragraph{Proof of \Thm{main-cont}:} In \Step{domain_reduce} of \Alg{tester} we set $k := (2 C \cdot \frac{d}{\eps})^{7}$ and sample a hypergrid $\bT = \prod_{i=1}^d T_i$, where each $T_i$ is formed by $k$ i.i.d. draws from $\cD_i$. By \Thm{cont_domain_reduction}, $\Expect_{\bT}[\eps_{f_{\bT}}] \geq \eps_f - \frac{C \cdot d}{k^{1/7}}$. Thus, if $\eps_f > \eps$, then $\Expect_{\bT}[\eps_{f_{\bT}}] \geq \eps / 2$. By \Clm{good-ell} there exists $\ell^\ast \in [L+1]$ such that $\Pr_{\bT}\left[\eps_{f_{\bT}} \geq \eps_{\ell^\ast}\right] \geq \frac{2^{\ell^{\ast}}\eps}{8 (\ell^{\ast})^2} \geq 4/Q_{\ell^{\ast}}$. Thus when $\ell$ is set to $\ell^{\ast}$ in \Alg{tester} \emph{at least one} of the $Q_{\ell^{\ast}}$ iterations of \Step{domain_reduce} returns $\bT$ satisfying $\eps_{f_{\bT}} \geq \eps_{\ell^{\ast}}$ with probability $\geq 1 - (1 - 4/Q_{\ell^{\ast}})^{Q_{\ell^{\ast}}} \geq 1 - (1/e)^4 \geq 15/16$. Thus, if $\eps_f > \eps$, then \Alg{tester} rejects with probability $> \frac{15}{16} \cdot \frac{2}{3} = 5/8$. On the other hand, if $f$ is monotone, then $f_{\bT}$ is always monotone and so \Alg{tester} accepts with probability $1$.

We now analyze the query complexity. Let $q(\eps,n,d)$ denote the query complexity of \pathtester~ with parameters $\eps, n$ and $d$. In particular, $q(\eps,k,d) \leq \tilde{O}(d^{5/6} \eps^{-4/3})$. 
Thus, the query complexity of \Alg{tester} is 

\begin{align*}
	\sum_{\ell=1}^{L+1} Q_\ell \cdot  q(\eps_{\ell},k,d) &= \sum_{\ell=1}^{L+1} \left\lceil\frac{32\ell^2}{2^{\ell}\eps}\right\rceil \cdot \tilde{O}\left({\frac{d^{5/6}}{2^{-4\ell/3}}}\right) = \tilde{O}\left(d^{5/6}\eps^{-1}\right)\sum_{\ell=1}^{L+1} \ell^2 \cdot \tilde{O}\left(2^{\ell/3}\right) \\
	&\leq \tilde{O}\left(d^{5/6}\eps^{-1}\right) L^3 \tilde{O}\left(2^{L/3}\right) \leq \tilde{O}\left(d^{5/6}\eps^{-4/3}\right)
\end{align*}

\noindent where in the last step we used the fact that $L = \Theta(\log (1/\eps))$. \pushQED{\qed} \qedhere \popQED 



\begin{claim} \label{clm:good-ell} If $\Expect_{\bT}[\eps_{f_{\bT}}] \geq \eps / 2$, then there exists $\ell^{\ast} \in [L+1]$ such that $\Pr\left[\eps_{f_{\bT}} \geq 2^{-\ell^{\ast}}\right] \geq \frac{2^{\ell^{\ast}}\eps}{8 (\ell^{\ast})^2}$. \end{claim}

\begin{proof} We have $\int_{0}^1 \Pr\left[\eps_{f_{\bT}} \geq t\right] dt = \Expect[\eps_{f_{\bT}}] \geq \eps/2$ and so $\int_{\eps/4}^1 \Pr\left[\eps_{f_{\bT}} \geq t\right] dt \geq \eps/4$. Thus,

\begin{align} \label{eq:levins}
	\frac{\eps}{4} \leq \int_{\eps/4}^1 \Pr\left[\eps_{f_{\bT}} \geq t\right] \leq \sum_{\ell=0}^L \int_{1/2^{\ell+1}}^{1/2^{\ell}} \Pr\left[\eps_{f_{\bT}} \geq t \right] dt \leq \sum_{\ell=0}^L \frac{1}{2^{\ell+1}} \Pr\left[\eps_{f_{\bT}} \geq 1/2^{\ell+1}\right]
    = \sum_{\ell=1}^{L+1} \frac{1}{2^{\ell}} \Pr\left[\eps_{f_{\bT}} \geq 1/2^{\ell}\right] \text{.}
\end{align}

\noindent For the sake of contradiction, assume $\Pr\left[\eps_{f_{\bT}} \geq 1/2^{\ell}\right] < \frac{2^{\ell}\eps}{8\ell^2}$ for all $\ell \in [L+1]$. Using \Eqn{levins}, we have

\begin{align*}
	\eps \leq 4\sum_{\ell=1}^{L+1} \frac{1}{2^{\ell}} \Pr\left[\eps_{f_{\bT}} \geq 1/2^{\ell}\right]
      < \frac{\eps}{2}\sum_{\ell=1}^{L+1} \frac{1}{\ell^2} < \frac{\eps}{2} \cdot \frac{\pi^2}{6} < \varepsilon \text{.}
\end{align*}

\noindent This is a contradiction. \end{proof}



\section{Lower Bound for Domain Reduction} \label{sec:lower}

In this section we prove the following lower bound for the number of uniform samples needed for a domain reduction result to hold for distance to monotonicity. Recall the domain reduction experiment for the hypergrid: given $f\colon [n]^d \to \{0,1\}$ and an integer $k \in \mathbb{Z}^+$, we choose $\bT := T_1 \times \cdots \times T_d$ where each $T_i$ is formed by taking $k$ i.i.d. uniform draws from $[n]$ with replacement. We then consider the restriction $f_{\bT}$.

\begin{theorem} [Lower Bound for Domain Reduction]\label{thm:lower_bound} There exists a function $f\colon [n]^d \to \{0,1\}$ 
with distance to monotonicity $\eps_f = \Omega(1)$, for which $\Expect_{\bT}[\eps_{f_{\bT}}] \leq O(k^2/d)$. In particular, $k = \Omega(\sqrt{d})$ samples in each dimension is necessary to preserve distance to monotonicity. 
\end{theorem}


\subsection{Proof of \Thm{lower_bound}}

We define the function $\nofan\colon [0,1]^d \to \{0,1\}$. The continuous domain is just a matter of convenience;
any $n$ that is a multiple of $d$ would suffice. It is easiest to think of $d$ individuals voting
for an outcome, where the $i$th vote $x_i$ is the "strength" of the vote. Based on their vote, an individual is labeled
as follows.

\begin{itemize} [noitemsep]
    \item $x_i \in [0,1-2/d]$: skeptic
    \item $x_i \in (1-2/d,1-1/d]$: supporter
    \item $x_i \in (1-1/d,1]$: fanatic
\end{itemize}

\medskip

$\nofan(x) = 1$ \emph{iff} there exists some individual who is a supporter.
The non-monotonicity is created by fanaticism.
If a unique supporter increases her vote to become a fanatic,
the function value can decrease.

\begin{claim} \label{clm:dist} The distance to monotonicity of $\nofan$ is $\Omega(1)$.
\end{claim}

\begin{proof} It is convenient to talk in terms of probability over the uniform distribution in $[0,1]^d$.
Define the following events, for $i \in [d]$.
\begin{itemize} [noitemsep]
    \item $\cS_i$: The $i$th individual is a supporter, and all others are skeptics.
    \item $\cF_i$: The $i$th individual is a fanatic, and all others are skeptics.
\end{itemize}

\medskip

Observe that all these events are disjoint. Also, $\Pr[\cS_i] = \Pr[\cF_i] = (1/d)(1-2/d)^{d-1} = \Omega(1/d)$.
Note that $\forall x \in \cS_i$, $\nofan(x) = 1$ and $\forall x \in \cF_i$, $\nofan(x) = 0$.

We construct a violation matching $M\colon \bigcup_i \cS_i \to \bigcup_i \cF_i$.
For $x \in \cS_i$, $M(x) = x + e_i/d$, where $e_i$ is the unit vector in dimension $i$.
For $x \in \cS_i$, $x_i \in (1-2/d,1-1/d]$, so $M(x)_i \in (1-1/d,1]$, and $M(x) \in \cF_i$. $M$ is a  bijection between $\cS_i$ and $\cF_i$, and all the $\cS_i, \cF_i$
sets are disjoint. Thus, $M$ is a violation matching. Since $\Pr\left[\bigcup_i \cS_i\right] = \Omega(d\cdot1/d)$, 
the distance to monotonicity is $\Omega(1)$.
\end{proof}


\begin{lemma} \label{lem:lb} Let $k \in \mathbb{Z}^+$ be any positive integer. If $\bT := T_1 \times \cdots \times T_d$ is a randomly chosen hypergrid, where for each $i \in [d]$, $T_i$ is a set formed by taking $k$ i.i.d. samples from the uniform distribution on $[0,1]$, then with probability $> 1-4k^2/d$, $\nofan_{\bT}$ is a monotone function.
\end{lemma}

\begin{proof} Each $T_i$ consists of $k$ u.a.r. elements in $[0,1]$. We can think
of each as a sampling of the $i$th individual's vote. For a fixed $i$,
let us upper bound the probability that $T_i$ contains strictly more than
one non-skeptic vote. 
This probability is
\begin{eqnarray*}
1- (1-2/d)^k - k(1-2/d)^{k-1}(2/d) & = & 1 - (1-2/d)^{k-1}(1-2/d + 2k/d) \\
& \leq & 1 - \left(1-\frac{2(k-1)}{d}\right)\left(1+\frac{2(k-1)}{d}\right) \leq 4k^2/d^2
\end{eqnarray*}

\noindent where we have used the bound $(1-x)^r \geq 1-xr$, for any $x \in [0,1]$ and $r \geq 1$. By the union bound over all dimensions, with probability $>1-4k^2/d$,
all $T_i$'s contain at most one non-skeptic vote. Consider $\nofan_{\bT}$,
some $x \in \bT$, and a dimension $i \in [d]$.
If the $i$th individual increases her vote (from $x$), there are three possibilities.
\begin{itemize} [noitemsep]
    \item The vote does not change. Then the function value does not change.
    \item The vote goes from a skeptic to a supporter. The function value can possibly increase, but not decrease.
    \item The vote goes from a skeptic to a fanatic. If $\nofan_{\bT}(x) = 1$, there must exist
    some $j \neq i$ that is a supporter. Thus, the function value remains $1$ regardless of $i$'s vote.
\end{itemize}

\medskip

\noindent In no case does the function value decrease.  Thus, $\nofan_{\bT}$ is monotone. \end{proof}

\Thm{lower_bound} follows from \Clm{dist} and \Lem{lb}.

\section{Domain Reduction for Variance}\label{sec:undirected}

In this section, we prove that, given $f\colon [n]^d \to \{0,1\}$, restricting $f$ to a random hypercube (domain reduction with $k=2$) suffices to preserve the \emph{variance} of $f$. 
Recall that the variance is defined $\var(f) := \Expect[f^2] - \Expect[f]^2$. In the proof, we will consider $f\colon [n]^d \to \{-1,1\}$ and so $\var(f) = 1 - \Expect[f]^2 = 1 - \widehat{f}(\emptyset)^2$. 


\begin{theorem} [Domain Reduction for Variance] \label{thm:undir_domain_reduction} Let $f\colon [n]^d \to \{0,1\}$ be any function. If $\bT := T_1 \times \cdots \times T_d$ is a randomly chosen sub-hypercube, where for each $i \in [d]$, $T_i$ is a (multi)-set formed by taking $2$ i.i.d. samples from the uniform distribution on $[n]$, then $\EX_{\bT}[\var(f_{\bT})] \geq \var(f)/2$. \end{theorem}


\begin{proof}  We will interpret $f$ as a Boolean function with $d\log n$ (Boolean) inputs,
	so $f\colon \{-1,1\}^{d\log n} \to \{-1,1\}$. We will index the inputs in $[d\log n]$, where
	the interval $I_i := [(i-1)\log n+1, i\log n]$ (the $i$th block) corresponds to the $i$th input
	in the original representation. Henceforth, $i$ will always index a block (and thereby, an input
	in the original representation). We use $x_j$ to denote the $j$th input bit.
	
	Let us think of the restriction in Boolean terms. Note that $f_{\bT}\colon \{-1,1\}^d \to \{-1,1\}$,
	and we use $y$ to denote an input to the restriction.
	In Boolean terms, $T_i$ picks two u.a.r. $\log n$ bit strings, 
	and forces the $i$th block of inputs, $I_i$, to be one of these. The choice between these
	is decided by $y_i$. Let us think of $T_i$ as follows. For every $j \in I_i$, it adds
	it to a set $R_i$ with probability $1/2$. All the inputs in $R_i$ will be fixed, while
	the inputs in $I_i\setminus R_i$ are alive (but correlated by $y_i$). Then, for every $j \in I_i$,
	it picks a u.a.r. bit $b_j$. (Call this string $B_i$.) 
	This is interpreted as follows. For every $j \in R_i$, $x_j$ is fixed to $b_j$.
	For every $j \in I_i\setminus R_i$, $x_j$ is set to $y_i b_j$.  
	The randomness of $T_i$ can therefore be represented
	as independently choosing $R_i$ and $B_i$.
	
	Consider some non-empty $S \subseteq I_i$. We have
	
	\begin{equation*}
	\prod_{j \in S} x_j = \prod_{j \in S \cap R_i} b_j \prod_{j \in S \setminus R_i} b_j y_i
	= y_i^{|S\setminus R_i|} \prod_{j \in S} b_j \text{.}
	\end{equation*}
	The expected value of the Fourier basis function is (as expected) zero. Recall that $S$ is non-empty and so
	\begin{equation} \label{eq:coeff}
	\EX_{T_i}\left[\EX_y\left[\prod_{j \in S} x_j\right]\right] = \EX_{R_i, B_i}\left[\EX_y\left[y_i^{|S\setminus R_i|} \prod_{j \in S} b_j\right]\right] 
	= \EX_{R_i}\left[\EX_y\left[y_i^{|S\setminus R_i|}\right]\right] \cdot \EX_{B_i}\left[\prod_{j \in S} b_j\right] = 0 \text{.}
	\end{equation}
	If $|S\setminus R_i|$ is even, then $\prod_{j \in S} x_j$
	is \emph{independent} of $y$. Then, $\EX_y\left[\prod_{j \in S} x_j\right]^2 = 1$.
	If $|S\setminus R_i|$ is odd, then $\prod_{j \in S} x_j$ is linear in $y_i$
	and $\EX_y\left[\prod_{j \in S} x_j\right] = 0$. Thus,
	\begin{equation} \label{eq:square}
	\EX_{T_i}\left[\EX_y\left[\prod_{j \in S} x_j\right]^2\right] = \Pr_{R_i}\left[\textrm{$|S\setminus R_i|$ is even}\right] = 1/2 \text{.}
	\end{equation}
	%
	
	\noindent Let us write out the Fourier expansion of $f$:
	$$    f(x)  = \sum_{S \subseteq [d\log n]} \widehat{f}(S) \cdot \chi_S(x) =  \sum_{\substack{S = S_1 \cup \ldots \cup S_d \\ \forall i, S_i \subseteq I_i}} \widehat{f}(S) \prod_{i \in [d]} \prod_{x_j \in S_i} x_j  \text{.}$$
	
	\noindent Let us write an expression for the square of the zeroth Fourier coefficient of the restriction:
	\begin{align} \label{eq:fourier_square}
		\EX_{\bT}\left[\widehat{f_{\bT}}(\emptyset)^2\right]  =  \EX_{\bT}\left[\left(\sum_{S \subseteq [d\log n]} \widehat{f}(S) \EX_y[\chi_S(x)]\right)^2\right] \text{.}
	\end{align}
	\noindent We stress that the choice of $x$ inside the expectations depend on $y$ (or $y'$) in the manner described
	before \Eqn{coeff}. Expanding the squared sum in \Eqn{fourier_square} and applying linearity of expectation, we get
	\begin{eqnarray} \label{eq:var_main}
	\EX_{\bT}\left[\widehat{f_{\bT}}(\emptyset)^2\right] & = & \EX_{\bT}\left[\sum_{S} \widehat{f}(S)^2 \EX_y\left[\chi_S(x)\right]^2 
	+ \sum_{S,T: S \neq T} \widehat{f}(S) \widehat{f}(T) \EX_y\left[\chi_S(x)\right] \EX_y\left[\chi_T(x)\right] \right] \nonumber \\
	& = & \sum_{S} \widehat{f}(S)^2 \EX_{\bT}\left[\EX_y[\chi_S(x)]^2\right]
	+ \sum_{S,T: S \neq T} \widehat{f}(S) \widehat{f}(T) \EX_{\bT}\left[\EX_y[\chi_S(x)] \EX_y[\chi_T(x)]\right] \text{.}
	\end{eqnarray}
	
	\noindent We will write $S = S_{i_1} \cup S_{i_2} \cdots \cup S_{i_k}$, where all $S_{i_r}$'s are non-empty.
	We deal with the first term of \Eqn{var_main}, using \Eqn{square} as follows:
	
	\begin{align} \label{eq:var_first_term}
		\EX_{\bT}\left[\EX_y[\chi_S(x)]^2\right] = \EX_{\bT}\left[\EX_y\left[\prod_{\ell \leq k} \prod_{j \in S_{i_\ell}} x_j\right]^2\right]
		= \prod_{\ell \leq k}\EX_{\bT}\left[\EX_y\left[\prod_{j \in S_{i_\ell}} x_j\right]^2\right] = 1/2^k \text{.}
	\end{align}
	The cross terms will be zero, using calculations analogous for \Eqn{coeff} (which is not directly used).
	We write $S = S_1 \cup \cdots \cup S_d$, where some of these may be empty. We deal with the second term of \Eqn{var_main} as follows:
	\begin{eqnarray} \label{eq:var_second_term}
	\EX_{\bT}[\EX_y[\chi_S(x)] \EX_y[\chi_T(x)]] & = & \EX_{\bT}\left[ \EX_y\left[\prod_{i \in [d]} \prod_{j \in S_i} x_j\right]
	\EX_y\left[\prod_{i \in [d]} \prod_{j \in T_i} x_j\right] \right] \nonumber \\
	& = & \prod_{i \in [d]} \EX_{R_i, B_i}\left[ \EX_{y_i}\left[ y^{|S_i\setminus R_i|}_i \prod_{j \in S_i} b_j\right]
	\EX_{y_i}\left[ y^{|T_i\setminus R_i|}_i \prod_{j \in T_i} b_j\right]\right] \nonumber \\
	& = & \prod_{i \in [d]} \EX_{R_i}\left[ \EX_{y_i}\left[ y^{|S_i\setminus R_i|}_i \right]
	\EX_{y_i}\left[ y^{|T_i\setminus R_i|}_i\right] \EX_{B_i} \left[\prod_{j \in S_i \Delta T_i} b_j\right]\right] \text{.}
	\end{eqnarray}
	
	
	
	There must exist some $i$ such that $S_i \Delta T_i \neq \emptyset$. For that $i$,
	$\EX_{B_i} \left[\prod_{j \in S_i \Delta T_i} b_j\right] = 0$, and thus for $S \neq T$,
	$\EX_{\bT}[\EX_y[\chi_S(x)] \EX_y[\chi_T(x)]] = 0$. Finally, plugging \Eqn{var_first_term} and \Eqn{var_second_term} into \Eqn{var_main} yields
	$$ \EX_{\bT}\left[\widehat{f_{\bT}}(\emptyset)^2\right] \leq \widehat{f}(\emptyset)^2 + \sum_{S \neq \emptyset} \widehat{f}(S)^2/2 
	= 1 - \var(f) + \var(f)/2 = 1-\var(f)/2 \text{.}$$
	Recall $\var(f_{\bT}) = 1-\widehat{f_{\bT}}(\emptyset)^2$. Thus, we rearrange to get $\EX_{\bT}[\var(f_{\bT})] = \EX_{\bT}\left[1-\widehat{f_{\bT}}(\emptyset)^2\right] \geq \var(f)/2$.
\end{proof} 

\subsection*{Acknowledgments}
We would like to thank the anonymous reviewers who have given constructive comments and pointed us to relevant material. In particular we would like to thank an anonymous reviewer who suggested the use of Levin's work investment strategy in \Sec{tester}.

\bibliography{derivative-testing}
\bibliographystyle{alpha}



\appendix

\end{document}